\documentclass[11pt]{article}
\usepackage{palatino}
\usepackage{StyleSheets/template}
\usepackage{StyleSheets/moreStyle}
\usepackage{StyleSheets/global}
\usepackage{subfiles}
\usepackage{url,color}
\usepackage{footnote}
\usepackage{threeparttable}
\usepackage{mdframed}
\usepackage{xspace}
\usepackage{stmaryrd}
\makesavenoteenv{tabular}
\makesavenoteenv{table}

\usepackage{mathdots}
\usetikzlibrary{matrix,decorations.pathreplacing,calc,patterns}

\newenvironment{construction}[1][]{%
\ifstrempty{#1}%
{\mdfsetup{%
frametitle={%
\tikz[baseline=(current bounding box.east),outer sep=0pt]
\node[anchor=east,rectangle,fill=black!20]
{\strut Construction};}}
}%
{\mdfsetup{%
frametitle={%
\tikz[baseline=(current bounding box.east),outer sep=0pt]
\node[anchor=east,rectangle,rounded corners,draw,fill=white]
{\strut~#1};}}%
}%
\mdfsetup{innertopmargin=5pt,linecolor=black,%
linewidth=.5pt,topline=true,
frametitleaboveskip=\dimexpr-\ht\strutbox\relax,}
\begin{mdframed}[]\relax%
}{\end{mdframed}}

\begin{document}

\title{Lossy Cryptography from Code-Based Assumptions\\
\vspace{.8em}
\large Dense-Sparse LPN: A New Subexponentially Hard LPN Variant in SZK}

\author{
   Quang Dao\thanks{Carnegie Mellon University. Email: \texttt{qvd@andrew.cmu.edu}} \and
   Aayush Jain\thanks{Carnegie Mellon University. Email: \texttt{aayushja@andrew.cmu.edu}}
}
\maketitle
 
\noteswarning 

\renewcommand{\subparagraph}[1]{\medskip\noindent\underline{\textit{#1}}}

\thispagestyle{empty}
\begin{abstract}

Over the past few decades, we have seen a proliferation of advanced cryptographic primitives with lossy or homomorphic properties built from various assumptions such as Quadratic Residuosity, Decisional Diffie-Hellman, and Learning with Errors. These primitives imply hard problems in the complexity class $\mathcal{SZK}$ (statistical zero-knowledge); as a consequence, they can only be based on assumptions that are broken in $\mathcal{BPP}^{\mathcal{SZK}}$. This poses a barrier for building advanced primitives from code-based assumptions, as the only known such assumption is Learning Parity with Noise (LPN) with an extremely low noise rate $\frac{\log^2 n}{n}$, which is broken in quasi-polynomial time.

In this work, we propose a new code-based assumption: Dense-Sparse LPN, that falls in the complexity class $\mathcal{BPP}^{\mathcal{SZK}}$ and is conjectured to be secure against subexponential time adversaries. Our assumption is a variant of LPN that is inspired by McEliece's cryptosystem and random $k\mbox{-}$XOR in average-case complexity. Roughly, the assumption states that
\[(\mathbf{T}\, \mathbf{M}, \mathbf{s} \,\mathbf{T}\, \mathbf{M} + \mathbf{e}) \quad \text{is indistinguishable from}\quad  (\mathbf{T} \,\mathbf{M}, \mathbf{u}),\]
for a random (dense) matrix $\mathbf{T}$, random sparse matrix $\mathbf{M}$, and sparse noise vector $\mathbf{e}$ drawn from the Bernoulli distribution with inverse polynomial noise probability.

We leverage our assumption to build lossy trapdoor functions (Peikert-Waters STOC 08). This gives the first post-quantum alternative to the lattice-based construction in the original paper. Lossy trapdoor functions, being a fundamental cryptographic tool, are known to enable a broad spectrum of both lossy and non-lossy cryptographic primitives; our construction thus implies these primitives in a generic manner. In particular, we achieve collision-resistant hash functions with plausible subexponential security, improving over a prior construction from LPN with noise rate $\frac{\log^2 n}{n}$ that is only quasi-polynomially secure.

\end{abstract}
\pagebreak
\tableofcontents
\thispagestyle{empty}
\pagebreak
\setcounter{page}{1}
\section{Introduction}\label{sec:intro}

Introduced in 2005, the Learning with Errors (LWE) assumption~\cite{STOC:Regev05} has emerged as a basis for designing post-quantum cryptography. LWE and its structured variants such as Ring-LWE~\cite{EC:LyuPeiReg10} or NTRU~\cite{HofPipSil98} are central to constructing a host of advanced cryptographic primitives including fully homomorphic encryption for classical~\cite{STOC:Gentry09,FOCS:BraVai11,C:GenSahWat13} and quantum computations~\cite{FOCS:Mahadev18b,C:Brakerski18}, attribute-based and other advanced encryption schemes~\cite{STOC:GorVaiWee13,C:GorVaiWee15}, non-interactive zero-knowledge~\cite{C:PeiShi19}, succinct arguments~\cite{FOCS:ChoJaiJin21}, and many other advances in classical~\cite{FOCS:GoyKopWat17,FOCS:WicZir17,STOC:GoyKopWat18,STOC:LinMooWic23} and quantum cryptography~\cite{FOCS:BCMVV18,FOCS:Mahadev18a}.

While LWE has proven to be surprisingly expressive in giving rise to advanced primitives, other post-quantum assumptions such as Learning Parity with Noise~\cite{C:BFKL93}, Isogenies~\cite{EPRINT:Couveignes06,EPRINT:RosSto06,AC:CLMPR18}, and Multivariate Quadratics~\cite{MQ-first-paper}, currently stand nowhere close in implying such advanced primitives, making LWE the single point of failure for designing advanced post-quantum cryptography. This state of affairs is highly unsatisfactory, since we would like to have some diversity in the assumptions implying a given primitive to hedge against unexpected cryptanalytic breakthroughs. Indeed, recent works~\cite{EC:CasDec23,EC:MMPPW23,EC:Robert23} have rendered the once-believed post-quantum assumption of SIDH classically broken in polynomial time.

This work aims to address a possible stagnation in terms of techniques and assumptions implying advanced post-quantum cryptography. This lack of versatility in assumptions for the most part can be attributed to the \emph{lack of techniques} in utilizing other post-quantum assumptions. The focus of this work lies in code-based cryptographic assumptions such as the Learning Parity with Noise (LPN) assumption~\cite{C:BFKL93} and its variants.

Learning Parity with Noise posits that random linear equations (with a planted secret solution) that is perturbed by sparse noise appears pseudorandom. Namely:
\begin{align*}
(\mat{A}, \mat{s} \cdot \mat{A}+\mat{e}) \approx_c  (\mat{A}, \mat{b}),
\end{align*}
where the coefficient matrix $\mat{A}$ is chosen at random from $\F^{n\times m}_{2}$, the secret $\mat{s}\leftarrow \F^{1\times n}_{2}$, $\mat{b}$ is chosen to be random vector in $\F^{1\times m}_{2}$ and the error vector $\mat{e} \in \F^{1\times m}_{2}$ is chosen so that each coordinate is i.i.d. sampled from the Bernoulli distribution with probability $\epsilon$. The problem is believed to be \emph{subexponentially secure}, meaning that subexponential $\exp(n^{O(1)})$-time adversaries have negligible distinguishing advantage when $\epsilon=O\left(\frac{1}{n^{\delta}}\right)$ for any constant $\delta \in (0,1)$.\footnote{A stronger version of subexponential security, which we do not consider in this work, also requires that the distinguishing advantage is an inverse subexponential $\exp(-n^{O(1)})$.} 

LPN is conceptually similar to LWE, in the sense that both posit pseudorandomness of planted random linear equations perturbed with noise. However, while for LWE the noise has low magnitude, for LPN it is sparse. One would expect that due to this similarity, LPN should imply a comparable variety of advanced primitives---yet this could not be any further from reality. On the one hand, recent works have leveraged the sparse noise structure to build specialized primitives such as homomorphic secret sharing~\cite{C:DIJL23}, or use LPN (over large fields) in combination with Bilinear Maps~\cite{STOC:MenVanOka91} and Goldreich's PRG~\cite{goldreich-prg} to build indistinguishability obfuscation~\cite{STOC:JaiLinSah21,EC:JaiLinSah22}.

On the other hand, despite almost three decades of research and in drastic contrast to LWE, we currently know only a handful of ways to leverage the LPN assumption. This is evident in the fact that aside from CPA/CCA secure public-key encryption schemes~\cite{FOCS:Alekhnovich03,AC:DotMulNas12,PKC:KilMasPie14,C:YuZha16} and UC-secure two-round oblivious transfer~\cite{EC:DGHMW20}, subexponentially secure variants of LPN alone are currently not known to imply any other primitives in Cryptomania~\cite{impagliazzo-five-worlds}.
Things seem to improve when one works with the quasi-polynomial time broken variant of the LPN assumption with very small noise probability $O(\frac{\log^2 n}{n})$, but even assuming this variant, very few primitives are known. These include collision-resistant and collapsing hash functions~\cite{EC:BLVW19,AC:YZWGL19,C:Zhandry22c}, identity-based encryption~\cite{EC:BLSV18}, and statistically-sender-private oblivious transfer~\cite{C:BitFre22}.

This brings us to our goal:

\begin{center}
{\bf Goal.} \emph{Devise new coding-theoretic techniques and assumptions for building advanced cryptography.}
\end{center}

To facilitate progress on the main goal above, we focus more on identifying properties of the assumption that could enable progress on the question, rather than focusing on specific primitives themselves. What makes assumptions such as LWE, Diffie-Hellman, Bilinear Maps, or Quadratic Reciprocity, special is that they can be used to design primitives with ``lossy'' or ``homomorphic'' properties, such as lossy trapdoor functions~\cite{STOC:PeiWat08} and linearly homomorphic encryption. Furthermore, the homomorphic/lossy properties of the assumption make them easier to work with to design other advanced Cryptomania primitives, such as attribute-based encryption or succinct arguments.

A key property that captures such assumptions is that they can be broken using an $\mathcal{SZK}$ oracle, where $\mathcal{SZK}$ is the complexity class of languages that have statistically-hiding zero-knowledge proofs. This ``$\mathcal{SZK}$-broken'' complexity class, known as $\BPP^{\mathcal{SZK}}$, consists of languages that can be decided efficiently using a \emph{statistical difference} ($\mathsf{SD}$) oracle~\cite{FOCS:SahVad97}. The $\mathsf{SD}$ oracle takes as input two polynomial sized distribution samplers $(\cD_0, \cD_1)$ (represented as randomized circuits), with the promise that either the statistical distance between the distributions is less than $\frac{1}{3}$ or it is more than $\frac{2}{3}$. The oracle then identifies which is the case.

This $\mathcal{SZK}$ regime indeed captures all of the assumptions mentioned above. For LPN, it is known \cite{EC:BLVW19} that the quasi-polynomial time broken variant with noise probability $O(\frac{\log^2 n}{n})$ can be broken with an $\mathcal{SZK}$ oracle, whereas subexponential time secure variants with inverse polynomial noise probability $\frac{1}{n^{\delta}}$ for $\delta\in (0,1)$ are currently not known to be in $\mathcal{SZK}$. This helps to explain why so little is known from LPN with inverse polynomial noise rate.

Therefore, to make progress in code-based cryptography, the first step would be to answer the following question:

\begin{center}
{\bf Question.} \emph{Is there a subexponentially-secure coding theoretic assumption with inverse polynomial noise probability that is in $\mathcal{SZK}$?}   
\end{center}

\subsection{Our Results}\label{sec:results}

We introduce a novel, well-motivated variant of LPN that we believe is secure against subexponential time algorithms. Unlike LPN where the matrix is chosen randomly, in our case it is a structured matrix. We work with a sufficiently small but inverse polynomial probability $\frac{1}{n^{\delta}}$ for a constant $\delta \in (0.5,1)$, a regime for which our assumption can be conjectured to be hard against subexponential $\exp(\tilde{O}(n^{1-\delta}))$-time adversaries.\footnote{However, an adversary's success probability in breaking our assumption is at least inverse quasi-polynomial; see~\Cref{sec:discussion-assumption} for more discussion.} Since our assumption implies primitives that can be broken by $\BPP^{\mathcal{SZK}}$,\footnote{It's a folklore result that lossy trapdoor functions, which we construct, lie in $\BPP^{\mathcal{SZK}}$; a formal proof can be found in~\cite[Appendix B]{TCC:FisRoh23}.}
our assumption indeed lies in $\BPP^{\mathcal{SZK}}$, making it the first plausible subexponential-time secure coding-theoretic assumption known to be in $\BPP^{\mathcal{SZK}}$.

\paragraph{New Assumption: Dense-Sparse LPN.} Our assumption borrows structural properties of two well-studied assumptions: the standard Learning Parity with Noise \cite{C:BFKL93} and the sparse Learning Parity with Noise problem~\cite{FOCS:Alekhnovich03}. Introduced in 2003 by Alekhnovich, Sparse LPN is exactly like standard LPN except that each column is chosen randomly among $k$-sparse vectors, where $k\geq 3$ is a constant. Sparse LPN is closely related to well-studied problems in the domain of constraint satisfaction and local pseudorandom generators~\cite{Gol00,CM01,STOC:Feige02,FOCS:MosShpTre03,FOCS:FeiKimOfe06,TCC:CEMT09,APPROX:BogQia09,STOC:AppBarWig10,TCC:AppBogRos12,BogdanovQ12,STOC:Applebaum12,SIAM:App13,CCC:OdoWit14,STOC:AppLov16,STOC:KMOW17,AC:CDMRR18,FOCS:AppKac19}, and when the number of samples satisfies $m= n^{\frac{k}{2}(1-\rho)}$ for any constant $0 < \rho < 1$,\footnote{When $m= \Omega(n^{k/2})$, due to the birthday bound two equations will repeat with constant probability, implying a trivial cheating strategy.}
it is believed to be subexponentially secure (provided the noise probability is a large enough inverse polynomial). Sparse LPN has also been shown to give rise to public-key encryption by Applebaum, Barak and Wigderson~\cite{STOC:AppBarWig10}, but not more advanced Cryptomania primitives.

Our assumption combines features from both LPN and Sparse LPN, and posits that LPN holds for the following \emph{Dense-Sparse} matrix distribution. We first sample a $k$-sparse matrix $\mat{M}\in \F^{n \times m}_{2}$ according to the distribution of coefficient matrix for the Sparse LPN assumption. We then sample a random (dense) matrix $\mat{T}\leftarrow \F^{n' \times n}_{2}$, where $n' = \alpha n$ for some constant $\alpha \in (0,1)$ (for simplicity, we set $\alpha=1/2$ in our paper). Finally, we give out $\mat{A}=\mat{T}\cdot \mat{M} \in \F^{n/2 \times m}_{2}$, and assume that random codewords of $\mat A$, perturbed by a Bernoulli noise vector $\mat e$ of inverse-polynomial noise rate, look pseudorandom. More formally, our assumption is stated as follows.

\begin{assumption}[Dense-Sparse LPN, informal]\label{ass:dense-sparse-lpn}
Let $k \ge 3$ be a constant, and consider parameters $n \in \N$, $m=m(n) < n^{k/2}$, and $\epsilon=\epsilon(n) < 1$. Let $\calM_{\sparse}$ be an efficiently sampleable ``good'' distribution over all $k$-sparse matrices in $\F_2^{n \times m}$.
We say that the \emph{$(n,m,k,\calM_{\sparse},\epsilon)$-Dense-Sparse LPN} assumption holds if the following two distributions are computationally indistinguishable:
\begin{align*}
    \{(\mat T \, \mat M, \mat s \, \mat T \, \mat M + \mat e)\}_{n \in \N} \approx_c \{(\mat T \, \mat M, \mat u)\}_{n \in \N},
\end{align*}
where $\mat T \gets \F_2^{n/2 \times n}$, $\mat M \gets \calM_{\sparse}$, $\mat s \gets \F_2^{1 \times n/2}$, $\mat e \gets \Ber(\epsilon)^{1 \times m}$, and $\mat u \gets \F_2^m$.
\end{assumption}
Looking ahead, our constructions will require us to assume Dense-Sparse LPN for an inverse polynomial noise rate $\epsilon = O(n^{-\delta})$ for some constant $\delta$ close to $1$, and the number of samples $m = \Omega(n^{1+\rho(\delta)})$ for a constant $\rho$ that depends on $\delta$. This parameter regime is plausibly secure against subexponential-time adversaries (see \Cref{sec:cryptanalysis} for details).

An important technical point in \Cref{ass:dense-sparse-lpn} is that of a ``good'' distribution of $k$-sparse matrices. This is due to the following reason: for $\mat M \in \F_2^{n \times m}$ chosen uniformly at random from the set of all $k$-sparse matrices, there is an \emph{inverse polynomial} probability of $O(m^2/n^k)$ that $\mat M$ has a vector $\mat x$ in its kernel of constant Hamming weight (so that $\mat M\, \mat x = \mat 0$). When this ``bad'' event happens, one cannot hope for distinguishing security to hold. Thus, since we want our distinguishing advantage to be negligible, we must sample $\mat M$ from another ``good'' distribution where this ``bad'' event happens with negligible probability; in particular, we will use the recent distribution constructed by Applebaum and Kachlon~\cite{FOCS:AppKac19}. We will expand on this in \Cref{sec:discussion-assumption}.

\paragraph{Connections to McEliece.} Our assumption can be viewed as applying the design principles of the classic McEliece~\cite{McEliece78} and Niederreiter~\cite{Niederreiter86} cryptosystems, which is to hide the sparse matrix $\mat M$ whose exposure would lead to an efficient attack in our parameter regime. In this sense, we follow a rich body of works on McEliece instantiated with different families of codes~\cite{mceliece-rm-codes,mceliece-convolutional-codes,BL05,wild-mceliece,PQCRYPTO:BerLanPet11,mceliece-ag-codes,EPRINT:MTSB12,mceliece-polar-codes,mceliece-polar-codes-2}. Nevertheless, there are two important distinctions between our assumption and prior McEliece variants. The first is that our variants are \emph{not} algebraically structured, unlike the original McEliece cryptosystem itself (which uses binary Goppa codes), or many other algebraic variants~\cite{mceliece-rm-codes,mceliece-ag-codes,BL05} or LDPC codes \cite{BC07,SCN:BalBodChi08} that have subsequently been broken~\cite{SS92-attack-mceliece-grs,BC07,EC:MinSho07,attack-mceliece-qc-codes,PQCRYPTO:Wieschebrink10,PQCRYPTO:LanTil13,EC:CouOtmTil14,PQCRYPTO:BCDOT16}.\footnote{In this sense, our assumption is related to the more combinatorial McEliece variant with \emph{medium-density parity check (MDPC)} codes~\cite{EPRINT:MTSB12}, which still remains secure to this day.} Secondly, we diverge from McEliece by making the masking matrix $\vec T$ \emph{compressing} of dimension $\alpha n \times n$ for any $\alpha < 1$, which is necessary for ensuring security in our setting. We expand more on this connection in \Cref{sec:crhf-overview}.

\paragraph{Cryptographic Applications.} We leverage Dense-Sparse LPN with inverse polynomial noise rate to build the following two primitives: a \emph{collision-resistant hash function} (in a simple and direct manner), and a \emph{lossy trapdoor function}.

Lossy Trapdoor Functions (LTDFs), introduced by Peikert and Waters in 2008~\cite{STOC:PeiWat08}, is a fundamental cryptographic tool that has found countless applications to building other cryptographic applications. LTDFs consist of a function family $F_{\fk}(\cdot )$ indexed by a public function key $\fk$, where the algorithm $\mathsf{Gen}$ that samples $\fk$ could sample keys in two modes. When the mode is injective, then the function $F_{\fk}(\cdot)$ is injective and can even be inverted uniquely using a trapdoor $\mathsf{td}$ generated by $\mathsf{Gen}$ at the same time of sampling $\fk$. In lossy mode, the range of the function $F_{\fk}(\cdot)$ is significantly smaller than the number of inputs. Equivalently, this also means that the conditional entropy in $x \in \{0,1\}^{\ell}$, given $y=F_{\fk}(x)$ for a random $x$ is large. In our setting we design such LTDFs for which the conditional entropy is at least $\Omega(\ell)$ where $\ell$ is the bit length of $x$. Finally, the two modes are required to be computationally indistinguishable, meaning that it is computationally hard to distinguish a random lossy key from a random injective key.

\begin{theorem}[informal]\label{thm:ltdf-informal}
Assuming Dense-Sparse LPN with inverse polynomial noise probability, there exists a construction of LTDF where the lossy mode loses any constant fraction $\Omega(\ell)$ of the input length $\ell$.

\end{theorem}

We give an example of how our parameters in \Cref{thm:ltdf-informal} can be concretely instantiated. In order for the lossy trapdoor function to lose (say) half of its entropy in lossy mode, we may set the sparsity parameter $k=6$, the number of samples $m=n^2$, and the error probability to be $\epsilon = O\left(n^{-10/11}\right)$. See \Cref{thm:ltdf} for the precise parameter regime required for LTDFs.

Lossy Trapdoor Functions are known from a number of quantum-broken assumptions such as Decisional Diffie-Hellman, Bilinear Maps, Quadratic Residuosity, Phi-Hiding, and Decisional Composite Residuosity (DCR). However, prior to our work, no post-quantum assumption barring LWE was known to imply lossy trapdoor functions, including LPN with noise probability $O(\frac{\log^2 n}{n})$ that is broken in quasi-polynomial time.

Since Lossy Trapdoor Functions are known to imply a number of lossy primitives, as a result we can generically realize those primitives from Dense-Sparse LPN. This list of primitives and applications include: collision-resistant hash functions and CCA secure encryption~\cite{STOC:PeiWat08}, dual-mode commitments and statistically-sender-private oblivious transfer~\cite{AC:HLOV11}, deterministic encryption~\cite{C:BolFehOne08}, trapdoor functions with many hardcore bits, analyzing OAEP \cite{C:KilOneSmi10}, hedged public-key encryption with bad randomness \cite{AC:BBNRSS09}, selective
opening security \cite{EC:BelHofYil09}, pseudo-entropy functions \cite{ICS:BraHasKal11}, point-function obfuscation \cite{C:Zhandry16}, computational extractors \cite{EC:DodVaiWic20,EC:GarKalKhu20}, incompressible encodings \cite{C:MorWic20}, and many more.

\subsection{Related Works}\label{sec:related-works}

\paragraph{LPN Variants and their Applications.} Recent works on \emph{pseudorandom correlation generators (PCGs)} \cite{CCS:BCGI18,C:BCGIKS19,C:BCGIKS20} and \emph{pseudorandom correlation functions (PCFs)}~\cite{FOCS:BCGIKS20} have proposed many novel variants of LPN with different matrix distributions~\cite{FOCS:BCGIKS20,C:CouRinRag21,C:BCGIKRS22,PKC:CouDuc23,C:RagRinTan23,C:BCCD23}. While these works are similar to ours in that they introduce new LPN variants, we introduce our variant (Dense-Sparse LPN) for a \emph{different} purpose: building more advanced lossy primitives in Cryptomania. In contrast, it is not known whether PCGs or PCFs imply public-key encryption or other Cryptomania primitives. 

\paragraph{Group Actions and SZK Primitives.} Besides lattices, certain assumptions on (suitable) group actions~\cite{C:BraYun90,EPRINT:Couveignes06} also imply primitives in $\mathcal{SZK}$, and are plausibly post-quantum secure against subexponential time adversaries. In more details, the authors of~\cite{AC:ADMP20} showed that group actions satisfying a \emph{weak pseudorandomness} property, such as those based on isogenies like CSIDH~\cite{AC:CLMPR18} or CSI-FiSh~\cite{AC:BeuKleVer19}, suffices for building a variety of $\mathcal{SZK}$ primitives such as hash proof system~\cite{EC:CraSho02}, dual-mode PKE~\cite{C:PeiVaiWat08}, malicious SSP-OT~\cite{SODA:NaoPin01,JC:HalKal12}, and more. However, their techniques do not seem to extend to building lossy trapdoor functions.

\paragraph{Future Directions.} A fascinating question left open in our work is whether we can construct similar Cryptomania primitives, though not known to generically follow from lossy trapdoor functions, from our Dense-Sparse LPN assumption. These primitives, which are known to be achievable from LPN with quasi-polynomial security, include laconic oblivious transfer~\cite{C:CDGGMP17,EC:BLSV18}, identity-based encryption~\cite{EC:BLSV18}, and (maliciously-secure) statistically-sender-private oblivious transfer~\cite{C:BitFre22}. In Section~\ref{sec:open-questions}, we sketch how our current assumption encounters roadblocks toward building these primitives.
\section{Overview of Techniques}\label{sec:tech-overview}

We now discuss the intuition for how the structural properties of Dense-Sparse LPN puts it in $\mathcal{SZK}$ and enables applications such as lossy trapdoor functions. We then discuss key points in the cryptanalysis of our assumption and end with some open questions.

\subsection{Collision-Resistant Hash Functions from Dense-Sparse LPN}\label{sec:crhf-overview}

\paragraph{Collision-Resistance from LPN.} In a nutshell, our progress on building lossy trapdoor functions from Dense-Sparse LPN is a result of achieving \emph{input compression} under a \emph{larger}, inverse polynomial noise rate. In both lattice-based and code-based cryptography, we consider the following hash function family:
\begin{align*}
h_{\mat{A}}(\mat x)= \mat{A} \cdot \mat{x},\qquad\text{ indexed by }\mat A \gets \calR^{n \times m}\text{ over some finite ring }\calR,
\end{align*}
and the input $\mat x$ comes from a ``low-norm'' distribution. In the case of LWE, we have $\mat A \gets \Z_q^{n \times m}$ for some modulus $q=q(n)$, and $\mat x \in \{0,1\}^m$. A simple calculation then shows that the number of inputs, which is $\abs{\{0,1\}^m} = 2^m$, is greater than the number of outputs, which is $\abs{\Z_q^n} =2^{n \log q}$, when $m > n \log q$. This is achievable even for exponential modulus $q = 2^{O(n)}$, for which we only need to set $m > O(n^2) = \poly(n)$.

In contrast, in the case of LPN, we have $\mat A \gets \F_2^{n \times m}$ and $\mat x \gets \calB(m,t)$, the Hamming ball of weight $t$ in $\F_2^m$. To achieve compression, we then need to ensure that
\begin{align}\label{eqn:compression-lpn}
    \text{number of inputs for }h_{\mat A} = \binom{m}{t} \quad\gg\quad 2^n = \text{number of outputs for }h_{\mat A}.
\end{align}
For $m = \poly(n)$, using the binomial approximations $\left(\frac{m}{t} \right)^t \le \binom{m}{t} \le \left(\frac{em}{t}\right)^t$, this is only possible when $t = \Omega(n / \log n)$. To see why this is problematic, we now sketch the proof (which can be found in~\cite{EC:BLVW19}) that $h_{\mat A}$ is collision-resistant, by a reduction to LPN with error probability $\epsilon=\epsilon(n)$. A collision $(\mat x_0, \mat x_1)$ for $h_{\mat A}$ gives us $\mat A \cdot (\mat x_0 - \mat x_1) = 0$, which implies that we have found a vector $\mat x = \mat x_0 - \mat x_1$ in the kernel of $\mat A$ that is at most $2t$-sparse. Such a vector can be used to detect bias in LPN samples $(\mat A, \mat b = \mat s \mat A + \mat e)$, since $\mat b \cdot \mat x = \mat e \cdot \mat x$ is distributed according to the Bernoulli distribution with error probability 
\begin{align*}
    \epsilon' \le \frac{1 - (1-\epsilon)^{2t}}{2},\quad\text{ and thus with bias }\quad 1-2\epsilon' = (1-\epsilon)^{2t} \ge 2^{-\Omega(\epsilon t)},
\end{align*}
by the \hyperref[lem:piling-up]{Piling-Up Lemma}. For this bias to be noticeable, compared to a uniform random bit $\mat b \cdot \mat x$ when $\mat b \gets \F_2^{m}$ is sampled randomly, we would need $\epsilon = O(\log n / t) = O(\log^2 n / n)$. With this error probability, LPN is broken in quasi-polynomial time $O(n^{\log n})$. More importantly, we cannot afford a lower error probability so that $\epsilon t = O(1)$, as LPN is fully broken with $\epsilon=O(\log n/n)$.\footnote{For $\epsilon=O(\log n /n)$, we can choose $n$ random coordinates of $\mat b = \mat s \mat A + \mat e$, which is error-free with noticeable probability, and then solve for $\mat s$.}

\paragraph{Achieving Higher Noise Rate with Sparse LPN.} Can we hope to achieve collision-resistance with larger error probability? Our key idea is that by changing the distribution of the matrix $\mat A$, we are able to \emph{reduce} the size of the output space, making compression possible at lower sparsity $t$ (and hence collision-resistance at higher noise rate $\epsilon$).
Indeed, if the matrix $\mat A$ is \emph{sparse} with each column having exactly $k$ ones, where $k$ is a constant or slightly super-constant, then the output space only consists of vectors $\mat y = \mat A \mat x \in \F_2^n$ that are at most $kt$-sparse. Thus, we achieve compression when
\begin{align*}
    \text{number of inputs for }h_{\mat A} = {m \choose t} \quad\gg \quad \sum_{s=0}^{kt} {n \choose s} = \text{number of outputs for }h_{\mat A}.
\end{align*}
Using binomial approximations and some straightforward calculations, the above is satisfied when
\begin{align*}
    \left(\frac{m}{t}\right)^t \gg kt \cdot \left(\frac{en}{kt}\right)^{kt} \quad\implies\quad t^{k-1} \gg \Omega\left( \frac{n^{k}}{m}\right).
\end{align*}
Thus, as soon as $t^{k-1}\gg \frac{n^{k}}{m}$, we will have $\mat{y}$ lose information on a random $t$-sparse $\mat{x}$. Let us try to understand how large $t$ needs to be. In the most conservative regime when $m=n^{1+\rho}$ for a small constant $\rho>0$, $t^{k-1}$ must be bigger than $n^{k-1-\rho}$, implying $t> n^{1-\frac{\rho}{k-1}}$ which approaches $n$ as $\rho$ approaches $0$. In the most aggressive setting when $m$ is close to $n^{\frac{k}{2}}$, $t^{k-1}$ must be bigger than $n^{\frac{k}{2}}$. This yields $t \approx \sqrt{n}$ if $k$ is a large enough constant.

More generally, if we wish to achieve a compression factor $D > 1$, meaning that the output length is a factor of $D$ smaller than the input length, then we would need
\begin{align}\label{eqn:compression-sparse-lpn}
    \left(\frac{m}{t}\right)^t > \left(kt \cdot \left(\frac{en}{kt}\right)^{kt}\right)^{D} \quad\implies\quad t^{D\cdot k-1} = \Omega \left(\frac{n^{D\cdot k}}{m}\right) .
\end{align}
We will refer to \Cref{eqn:compression-sparse-lpn} as the \emph{compression} equation. Similar to the above estimate, we can have $t=n^{\delta}$ be polynomially smaller than $n$ for a polynomial number of samples $m=n^{1+\rho_{k,D}(\delta)}$, where $\rho_{k,D}(\delta)$ is a constant related to $\delta$. This implies that collision-resistance can be achieved at an inverse polynomial error probability $\epsilon = O(1/t) = O(1/n^\delta)$. We call this the \emph{compression regime} of Sparse LPN.

\paragraph{Sparse LPN in its Compression Regime is Broken.}
Unfortunately, while Sparse LPN in the compression regime would imply collision-resistant hash functions, it is not secure. We give a new but simple attack, that in the compression regime with any factor $D > 1$, one can easily find $O(t)$ sparse vectors $\mat v \in \F_2^{m}$ so that $\mat{A}\mat{v}=\mat{0}$. Thus, unlike (dense) LPN, Sparse LPN can be broken in polynomial time even at sufficiently small but inverse-polynomial noise probability $\epsilon = O(n^{-\delta})$, and with (related) polynomial number of samples $m = n^{1+\rho_{k,D}(\delta)}$.

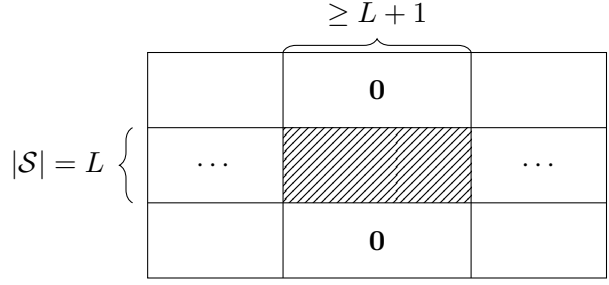
\begin{figure}
    \centering
    
    \begin{tikzpicture}
        \newcommand{\narrowcellwidth}{1.8cm}
        \newcommand{\widecellwidth}{2.5cm}
        \newcommand{\cellheight}{1cm}
    
        \draw (0,0) rectangle (2*\narrowcellwidth+\widecellwidth,3*\cellheight);
    
        \draw (\narrowcellwidth,0) -- (\narrowcellwidth,3*\cellheight);
        \draw (\narrowcellwidth+\widecellwidth,0) -- (\narrowcellwidth+\widecellwidth,3*\cellheight);
        \draw (0,\cellheight) -- (2*\narrowcellwidth+\widecellwidth,\cellheight);
        \draw (0,2*\cellheight) -- (2*\narrowcellwidth+\widecellwidth,2*\cellheight);
    
        \draw[decoration={brace,amplitude=5pt},decorate] (-0.2,\cellheight) -- (-0.2,2*\cellheight) node[midway,xshift=-1cm] {$|\calS| = L$};
    
        \draw[decoration={brace,amplitude=5pt},decorate] (\narrowcellwidth,3*\cellheight+0.2) -- (\narrowcellwidth+\widecellwidth,3*\cellheight+0.2) node[midway,yshift=0.5cm] {$\geq L+1$};
    
        \fill[pattern=north east lines] (\narrowcellwidth,\cellheight) rectangle (\narrowcellwidth+\widecellwidth,2*\cellheight);
    
        \node at (\narrowcellwidth+0.5*\widecellwidth,0.5*\cellheight) {$\mat 0$};
        \node at (\narrowcellwidth+0.5*\widecellwidth,2.5*\cellheight) {$\mat 0$};
        
        \node at (0.5*\narrowcellwidth,1.5*\cellheight) {$\cdots$}; 
        \node at (2*\narrowcellwidth+0.65*\widecellwidth,1.5*\cellheight) {$\cdots$}; 

    \end{tikzpicture}
    
    \caption{Attack against Sparse LPN with sparsity $k$, number of samples $m=\Omega(n \cdot (n/t)^{k-1})$, and noise rate $\epsilon = O(\log n/t)$. We focus on a set of rows $\calS$ of size $L = O(t)$, and find all columns that is non-zero only within rows contained in $T$. If we find at least $L+1$ such columns, then we can find a linear dependence between these columns and come up with a $\le L$-sparse vector $\mat x$ such that $\mat A \mat x = \mat 0$, which can be used to detect noticeable bias in the Sparse LPN samples.}
    \label{fig:sparse-lpn-attack}
\end{figure}

The attack can be described as follows. For simplicity, assume that each column of $\mat{A}$ is randomly and independently chosen from the set of all $k$-sparse columns. We want to find a set $\calS \subset [n]$ of size $L$ so that there are $L+1$ columns $\{ \mat{a}_1,\ldots,\mat{a}_{L+1}\}$ that are supported entirely in $\calS$. Namely, these columns take the value $0$ for indices in $[n]\setminus \calS$. Once we have found these columns, we can easily find some non-zero combinations combining $\{\mat{a}_1,\ldots,\mat{a}_{L+1}\}$ that sum to $0$ since they must be linearly dependent.

For what $L$ can we expect this to happen? We compute the probability that for an arbitrary set $S$ of size $L$ there exist $L+1$ column vectors supported entirely in $S$. Since the probability that a single vector is supported in $S$ is ${L\choose k}/{n \choose k}$, the expected number of such equations becomes roughly $m\cdot {L\choose k}/{n \choose k}$. Thus, for $m$ samples, we need to set $L$ so that
\begin{align}
m\cdot \frac{{L\choose k}}{{n \choose k}} \approx m \cdot \left(\frac{L}{n}\right)^k \gg L,\qquad\text{ which is satisfied when }L^{k-1} \approx \frac{n^{k}}{m}.
\end{align}
Therefore, $L$ is up to a constant multiple of $t$ satisfying the compression equation $t^{k-1}=\Omega(\frac{n^{k}}{m})$.

\paragraph{McEliece-style Wrapper to the Rescue.} Note that the above attack crucially relies on being able to identify the support, or locality pattern, of the column vectors. Thus, if we can mask these locality patterns, such as by applying a random linear transformation to the columns, then we can prevent our attack.

This idea of masking a matrix, whose exposure would lead to an efficient attack, goes back to the McEliece cryptosystem~\cite{McEliece78}, and is the motivation behind Dense-Sparse LPN. Recall that in McEliece and its variants, the public key consists of a ``randomized'' code generating matrix $\mat A = \mat S \mat G \mat P$, hiding a ``nice'' representation $\mat G \in \F_2^{n \times m}$ (such as a binary Goppa code) that enables efficient decoding. Here the masking matrices consist of a random \emph{square} matrix $\mat{S} \in \F_2^{n \times n}$ multiplied on the left and a permutation matrix $\mat{P} \in \F_2^{m \times m}$ multiplied on the right. The ciphertexts are then LPN samples with respect to this structured matrix.

Our Dense-Sparse LPN assumption makes several modifications to the McEliece template to suit our specific needs. Instead of an algebraically structured matrix $\mat G$, we work with $k$-sparse matrices $\mat{M}$ which do not have any algebraic structure. We also omit the permutation matrix $\mat{P}$ as the sparse matrix distribution can be made invariant under column permutations. Finally, we multiply with a random \emph{compressing} matrix $\mat T \in \F_2^{\frac{n}{2} \times n}$ to get $\mat A = \mat T \mat M$. This change serves two purposes: it ensures security, as a square matrix $\mat S \in \F_2^{n \times n}$ does not adequately mask $\mat M$ (see \Cref{sec:cryptanalysis} for the attack), and it aligns with our goal of using $\mat M$ for its lossy properties rather than for decoding.

We discuss further aspects of Dense-Sparse LPN in \Cref{sec:discussion-assumption}; for the moment, we shall go back to constructing collision-resistant hashes from our assumption.

\paragraph{Back to Collision-Resistance.} We now re-examine our hash function $h_{\mat{A}} = \mat A \mat x = \mat T \mat M \mat x$, which takes $t$-sparse inputs $\mat x \in \F^{m}_{2}$ and maps them to $\F^{\frac{n}{2}}_{2}$. When $n,m,t$ is chosen in the compression regime to satisfy \Cref{eqn:compression-sparse-lpn} for any factor $D > 1$, the mapping $h_{\mat M}$ sending $\mat x \mapsto \mat x' = \mat{M}\mat{x}$ admits collisions. Since $\mat x' \in \F_2^n$ is at most $kt$-sparse, which is lower than the threshold $O(n/\log n)$ for $\mat T$ to get collisions, it follows that, with overwhelming probability, all collisions of $h_{\mat{A}}$ come from collisions of $h_{\mat M}$.

Unfortunately, $h_{\mat A}$ is no longer compressing. Since the output of $h_{\mat A}$ is no longer sparse, we can only bound the output size by its length $\frac{n}{2}$, which is more than the input length of $\log_2 {m \choose t}=O(t\log n)$ as $t$ is polynomially smaller than $n$. 
To get around this issue, we multiply the function $h_{\mat A}$ with another compressing matrix $\mat{U}\in \F^{\ell \times \frac{n}{2}}_{2}$ for a suitable $\ell$. Equivalently, we now consider the hash family
\begin{align*}
    h_{\mat A'}(\mat x) = \mat A' \cdot \mat x \in \F_2^{\ell},\qquad\text{where }\quad \mat{A}'=\mat U\cdot \mat A = (\mat U \cdot \mat{T}) \cdot \mat{M} \in \F_2^{\ell \times m}.
\end{align*}
Note that $\mat V = \mat U \cdot \mat T \in \F_2^{\ell \times n}$ is identically distributed to a random matrix of the same dimensions. We will set $\ell$ so that two conditions are satisfied:
\begin{itemize}
    \item To ensure $h_{\mat A'}$ is compressing, we need $2^\ell \ll {m\choose t}$.
    \item To ensure $h_{\mat A'}$ is collision-resistant, we want the linear transformation $\mat y \mapsto \mat{V}\mat{y}$ to map different vectors $\mat y= \mat M \mat x$ that are at most $kt$-sparse to different elements of $\F_2^{\ell}$. This is satisfied with overwhelming probability if $2^\ell \gg {n \choose kt}^2 \approx \abs{\calB_{\le}(n,kt)}^2$ by a birthday bound, where $\calB_{\le}(n,kt)$ denotes the Hamming ball of radius $kt$.
\end{itemize}
We can meet both conditions by setting the compression factor $D > 2$, so that ${m\choose t} \gg {n \choose k \cdot t}^{2}$. We give a formal argument in \Cref{sec:CRHF}, where collision-resistance can be reduced to Dense-Sparse LPN with an inverse polynomial error probability $\epsilon = O(n^{-\delta})$. Finally, we note that while our later construction of lossy trapdoor functions generically implies collision-resistant hash functions~\cite{STOC:PeiWat08}, the hash construction we just sketched is more direct and conceptually simpler.

\subsection{Lossy Trapdoor Functions}\label{sec:ltdf-overview}

We now describe the main ideas behind our lossy trapdoor function construction, which on a high level builds upon the lattice-based template of~\cite{STOC:PeiWat08}. The core of our contribution lies in identifying that, for the same compression regime of Dense-Sparse LPN as sketched above, we can ensure both lossiness and invertibility depending on the mode being instantiated.

Our function key is of the following form. Given a matrix $\mat A = \mat T \mat M \in \F_2^{\frac{n}{2} \times m}$ drawn from the Dense-Sparse LPN matrix distribution, we generate samples $(\mat b_i=\mat s_i \mat A + \mat e_i)_{i=1}^\ell$ for a parameter $\ell=\ell(n)$ to be chosen later. Equivalently, we may concatenate these samples to get a matrix $\mat B = \mat S \mat A + \mat E \in \F_2^{\ell \times m}$. This matrix $\mat B$ will be given out as-is for the lossy mode, and for the injective mode, will be used to hide a \emph{(robust) compressed sensing} instance $\mat C \in \F_2^{\ell \times m}$ which allows for efficient inversion. To summarize, the function key is as follows:
\begin{align*}
    \fk = \begin{cases}
        (\mat A, \mat B = \mat S \mat A + \mat E) & \text{ if }\mode=\loss,\\
        (\mat A, \mat B = \mat S \mat A + \mat E + \mat C) & \text{ if }\mode=\inj,\text{ with }\td=\mat S.
    \end{cases}
\end{align*}
From the Dense-Sparse LPN assumption, it is clear that the two modes are computationally indistinguishable. Now, the input to the function will be $t$-sparse vectors $\mat x \in \F_2^{m}$, and function evaluation returns 
\begin{align*}
    \mat y = (\mat y_1,\mat y_2),\qquad\text{ where }\mat y_1 = \mat A \mat x \in \F_2^{\frac{n}{2}},\text{ and }\mat y_2 = \mat B \mat x \in \F_2^{\ell}.
\end{align*}
To invert in injective mode (with trapdoor $\mat S$), we compute $\mat y' = \mat y_2 - \mat S \mat y_1 = \mat C \mat x + \mat E \mat x$, then use the decoding guarantee of the compressed sensing matrix $\mat C$ to recover $\mat x$ in the presence of noise $\mat E \mat x$.

We now analyze the two modes to figure out the parameters that would ensure both lossiness and efficient inversion holds with $1-\negl(n)$ probability over the choice of the function key.

\paragraph{Lossy Mode.} We want to choose parameters so that both $\mat y_1 = \mat A \mat x$ and $\mat y_2 = (\mat S \mat A + \mat E) \mat x$ loses information about $\mat x$. We will reason separately about these two components as follows:
\begin{itemize}
    \item Using the decomposition $\mat A = \mat T \mat M$ for a random (dense) $\mat T$ and sparse $\mat M$, it suffices to have $\mat x' = \mat{M}\mat{x}$ lose information about $\mat{x}$. This is satisfied in the compression regime for Sparse LPN where \Cref{eqn:compression-sparse-lpn} requires $t = \Omega\left(n\cdot (n/m)^{\frac{1}{Dk-1}}\right)$ for a compression factor $D>1$.
    
    \item For a fixed value $\mat y_1$, we can see that $\mat y_2 = \mat S \mat y_1 + \mat E \mat x$ lies in a Hamming ball of radius $\norm{\mat E \mat x}_0$ around $\mat S \mat y_1$. If we work with error probability $\epsilon$, then $\mat E \gets \Ber(\epsilon)^{\ell \times m}$ and $\norm{\mat x}_0 = t$ implies that $\mat E \mat x \sim \Ber(\epsilon')^{\ell}$, where $\epsilon' = \frac{1-(1-2\epsilon)^t}{2} \le \epsilon t$. We want the size of this ball, which is roughly ${\ell \choose \epsilon t \ell}$, to be at most a $(D')^{th}$-root of the number of inputs which is ${m \choose t}$.
\end{itemize}
Putting things together, we have that
\begin{align}
    \abs{\{(\mat A \mat x, \mat B \mat x) \mid \norm{\mat x}_0 = t\}} < {m \choose t}^{1/D} \cdot {m \choose t}^{1/D'} = {m \choose t}^{1/D + 1/D'}.
\end{align}
Therefore, we can ensure that the output length is an arbitrarily small constant of the input length by setting $D,D'$ to be large enough, with error probability $\epsilon = O(1/t)$ and $\ell = \Theta(\log {m \choose t}) = \Theta(t \log n)$.

\paragraph{Injective Mode.} Do the parameters required for lossy mode also enable efficient inversion? To answer this, we need to design a matrix $\mat C \in \F_2^{\ell \times m}$ equipped with an efficient decoding algorithm that can recover a $t$-sparse vector $\mat x$ from $\mat y' = \mat C \mat x + \mat e \in \F_2^{\ell}$, where $\mat e = \mat E \mat x$ is a noise term that is constant-fraction sparse with overwhelming probability.

If our task were to recover a (dense) vector $\mat x$, then we can simply pick $\mat C$ to be (the transpose of) an error-correcting code. Then $\mat C \mat x$ is (the transpose of) a codeword, which is then perturbed in a constant number of entries to form $\vec y'$. Using an efficient decoding algorithm for $\mat C$ that can correct a constant fraction of errors, we can recover $\mat x$ from $\mat y'$.

However, in our case we have to recover a $t$-sparse vector $\mat x \in \F_2^{m}$. Our idea is to restrict the inversion process to only a \emph{special} subset of such $t$-sparse vectors, namely the ones that arise as the result of \emph{sparsifying} dense vectors $\mat z \in \F_2^{t \log(\frac{m}{t})}$. The sparsification process is as follows:
\begin{align*}
    \mat{z} = [\underbrace{\mat{z}_1}_{\log(\frac{m}{t})} \| \dots \| \underbrace{\mat{z}_t}_{\log(\frac{m}{t})} ] \quad\implies \quad 
    \spfy(\mat z) = [ \underbrace{(0,\dots,1,\dots,0)}_{\mat z_1\text{-th position}} \| \dots \| \underbrace{(0,\dots,1,\dots,0)}_{\mat z_t\text{-th position}}] \in \F_2^{m},
\end{align*}
where we interpret $\mat z_i \in \F_2^{\log(\frac{m}{t})}$ as a number in $\{0,\dots,\frac{m}{t}-1\}$ for all $i \in [t]$. This gives a bijection between binary vectors of length $t\log(\frac{m}{t})$, with \emph{regular} $t$-sparse vectors of length $m$.
We may also recover $\mat z$ from $\spfy(\mat z)$ by multiplying with a \emph{gadget matrix} $\mat G \in \F_2^{t\log(\frac{m}{t}) \times m}$ (whose formula can be found in \Cref{def:gadget-matrix}). In other words, we have $\mat G \cdot \spfy(\mat z) = \mat z$ for all $\mat z \in \F_2^{t\log(\frac{m}{t})}$. 
Note that these procedures have appeared in prior works~\cite{EC:BLVW19,AC:YZWGL19}, and can be seen as a code-based analogue of binary decomposition and the gadget matrix in lattice-based cryptography~\cite{EC:MicPei12}.

Given these tools, we can instantiate the injective mode as follows. We first slightly change the input space of our function to be regular $t$-sparse vectors $\mat x \in \F_2^m$, which are in bijection with $\mat z = \mat G \mat x \in \F_2^{t\log(\frac{m}{t})}$. We now set 
\begin{align*}
    \mat C = \mat C' \cdot \mat G,
\end{align*}
where $\mat C' \in \F_2^{\ell \times t\log(\frac{m}{t})}$ is the transpose of an error-correcting code with constant rate and distance~\cite{justesen-code}. Therefore, using the decoding of $\mat C'$ we may efficiently recover $\mat z$, and hence $\mat x=\spfy(\mat z)$, from $\mat y' = \mat C \mat x + \mat e = \mat C' \mat G \mat x + \mat e = \mat C' \mat z + \mat e$.

\paragraph{All-but-one LTDFs.} Our LTDF construction also generalize straightforwardly to realizing its \emph{all-but-one (ABO)} variant. In an ABO-LTDF, we have an exponential of branches $\calB=\F_2^{L}$ where one distinguished branch $b^*$ is lossy, and all other branches are injective. Additionally, given the function key $\fk$ it is difficult to find out which branch is distinguished. While there is a generic transformation from LTDFs to its all-but-one version~\cite{STOC:PeiWat08}, it presents a tradeoff between the number of branches supported and the degradation of the lossiness parameter.

We avoid this tradeoff by leveraging algebraic properties in our setting. In particular, we achieve ABO-LTDFs with number of branches $2^{L}$, where $L=t\log(\frac{m}{t})$, through a simple twist on our construction above. At a high level, we associate the branches $b \in \F_2^{t\log(\frac{m}{t})}$ with a matrix family $\mat H_b \in \F_2^{L \times L}$ such that $\mat H_b - \mat H_{b'}$ is invertible for all $b \ne b'$. Such a \emph{full-rank difference (FRD)} family of matrices has appeared in prior works~\cite{EC:AgrBonBoy10,PKC:KilMasPie14}, and we refer to \Cref{sec:ltdf} for the details of our ABO construction.

\paragraph{Why Low-Noise LPN does not Suffice.} As a final remark, let us argue why LPN with noise probability $\epsilon = O(\log^2 n / n)$ does not imply lossy trapdoor functions from our template above. The reason is that since $\epsilon t = O(\log n)$, the vector $\mat E \mat x$, with $\mat E \gets \Ber(\epsilon)^{\ell \times m}$ and $\norm{\mat x}_0 = t$, is Bernoulli distributed with error $\epsilon' = \frac{1}{2} - \frac{1}{\poly(n)}$. This requires setting $\ell = n^{1+\Omega(1)}$ for successful inversion in injective mode.\footnote{Error-correcting codes of relative distance $\frac{1}{2} - \frac{1}{\poly(n)}$ must have rate at most $1/\poly(n)$, so we have $\ell = n^{1+\Omega(1)}$.} However, such a large $\ell$ prevents any hope of achieving lossiness, as the Hamming ball around each $\mat y_1$ is simply too large:
\begin{align}
    {\ell \choose \epsilon' \ell} = 2^{n^{1+\Omega(1)}} \gg 2^{O(n)} = {m \choose t}.
\end{align}

\subsection{Discussion on Our Assumption}\label{sec:discussion-assumption}

We now discuss several aspects of our Dense-Sparse LPN assumption. Since our assumption combines aspects of both standard LPN and Sparse LPN, it also inherits some subtle considerations that arise in the context of choosing ``good'' matrices for Sparse LPN.

\paragraph{Sampling ``Good'' Matrices for Sparse LPN.}
In the Sparse LPN assumption, the matrix $\mat{M}$ is chosen in $\F^{n\times m}_{2}$, where $m \ll n^{\frac{k}{2}}$, such that every column is $k$-sparse where $k$ is some constant. For this parameter setting, there is an inverse polynomial probability of $\mat{M}$ having a constant-sparse vector $\mat x$ in its kernel, so that $\mat M \mat x = \mat 0$.\footnote{For instance, $\mat M$ has two identical columns with probability $O(\frac{m^2}{n^k})$.} When this ``bad'' event happens, an adversary can find $\mat x$ in polynomial time and distinguish Sparse LPN samples from random.

However, this ``bad'' event is only over the choice of matrix $\mat M$. Outside of this bad event, it is known that if $m = \tilde{O}(n^{1+(\frac{k}{2}-1)(1-\rho)})$ for some $\rho >0$, then with overwhelming probability, the minimum Hamming weight for a vector $\mat x$ that satisfies $\mat M \mat x = \mat 0$ is at least $O(n^{\rho})$ (see \Cref{lem:large-dual-distance} for details). Therefore, if we consider Sparse LPN with a large enough noise probability $\epsilon = O(n^{-\rho})$, giving an adversary subexponential time do not seem to increase its advantage beyond the probability of sampling a ``bad'' matrix.

We note that this issue is present in all prior cryptographic constructions relying on Sparse LPN. Applebaum, Barak and Wigderson~\cite{STOC:AppBarWig10} resolved this issue by weakening the indistinguishability advantage to be only $o(1)$; this suffices for their application of building public-key encryption, as they could rely on security amplification~\cite{C:HolRen05}. However, our goal is to build lossy primitives, which do not amplify well~\cite{TCC:PieRosSeg12}, so we must rely on a sampling $\mat M$ from a distribution that avoids sampling a ``bad'' matrix with overwhelming probability. We will use the following efficiently sampleable distribution from~\cite{FOCS:AppKac19} which has this property.

\begin{theorem}[informal]
    For every even $k \ge 6$, every $1 < c < k/4$, and every $0 < \gamma < k - 4c$, there exists an efficiently sampleable distribution of $k$-sparse matrices $\mat M \in \F_2^{n \times n^c}$ such that with overwhelming probability, every nonzero vector $\mat x$ in the kernel of $\mat M$ has Hamming weight at least $O(n^{\delta})$, where $\delta = \frac{k-4c-\gamma}{k-\gamma-4}$.
\end{theorem}

While this distribution does not sample matrix from all possible parameter regime of Sparse LPN (for instance, we have $m = n^c < n^{k/4}$), it suffices to instantiate the compression regime in \Cref{eqn:compression-sparse-lpn} with any compression factor $D>1$, where the noise probability remains inverse polynomial and we have plausible security against subexponential time attacks. Therefore, we can use this distribution to instantiate the sparse matrix $\mat M$ in our Dense-Sparse LPN assumption.

\paragraph{Almost-Subexponential Regime.} Alternatively, we can work with a $k$ that is mildly super-constant, such as $k = \log \log n$. In this setting, $m=n^c$ can be an arbitrary polynomial in $n$, and the ``bad'' matrix probability is now negligible, namely $n^{-O(\log \log n)}$. This allows us to sidestep the need for special distributions, and just sample $\mat M$ uniformly at random from the set of $k$-sparse matrices. \Cref{eqn:compression-sparse-lpn} now implies that we can achieve compression factor $D>1$ when
\begin{align}
    t = O\left( \left(\frac{n^{Dk}}{m}\right)^{\frac{1}{Dk-1}} \right) = n^{1 - O\left(\frac{1}{\log \log n}\right)}.
\end{align}
Thus, in our construction of LTDFs, we can rely on Dense-Sparse LPN with error probability $\epsilon = O(1/t) = n^{-O\left(\frac{1}{\log \log n}\right)}$ which is \emph{inverse almost-polynomial}. Note that this noise probability is still larger than the $O(\frac{\log^2 n}{n})$ noise rate of (standard) LPN. Therefore, our Dense-Sparse LPN assumption with this error rate is plausibly secure against \emph{almost-subexponential} $2^{n^{O\left(\frac{1}{\log \log n}\right)}}$-time attacks, unlike the quasi-polynomial security of (standard) LPN in its compression regime.

\paragraph{A Stronger Variant of Our Assumption.} Finally, we can consider making a stronger assumption that together with (standard) LPN implies Dense-Sparse LPN. In this variant, we now assume that the matrix $\mat A = \mat T \mat M $ from the Dense-Sparse matrix distribution is actually \emph{pseudorandom}.

\begin{assumption}[Dense-Sparse Indistinguishability]\label{ass:dense-sparse-pseudorandom}
For $\mat T \gets \F_2^{\frac{n}{2} \times n}$, $\mat M \gets \F_2^{n \times m}$ drawn from a \hyperref[def:good-dist]{good} distribution of $k$-sparse matrices, and $\mat U \gets \F_2^{\frac{n}{2} \times m}$, we have $\mat{T}\mat{M}\approx_c \mat{U}$.
\end{assumption}
By a simple hybrid argument, \Cref{ass:dense-sparse-pseudorandom} together with (standard) LPN indeed implies Dense-Sparse LPN with noise rate as small as $O(\frac{\log^2 n}{n})$. We do not see any better attack against this assumption than against Dense-Sparse LPN, and refer to \Cref{sec:cryptanalysis} for initial cryptanalysis regarding this assumption.

\subsection{Open Questions}\label{sec:open-questions}

At a high level, our Dense-Sparse LPN assumption is designed specifically so that we can harness the lossiness property of the $k$-sparse matrix $\mat{M}$, and do so without losing security by applying a random compressing matrix $\mat{T}$ before giving out $\mat A = \mat{T}\mat{M}$. In our construction of collision-resistant hash functions (CRHFs), we even managed to convert this lossiness into actual compression of the output. However, it seems difficult to achieve compression directly without giving up on other properties of the assumption.

Let us explain this difficulty by attempting to construct a (single-server) private information retrieval (PIR) scheme~\cite{FOCS:CGKS95,FOCS:KusOst97}.\footnote{It is not known whether lossy trapdoor functions imply PIR, nor is there a black-box separation between two primitives.} In PIR, we have a server holding a database $\mat D \in \F_2^n$, and a client who wishes to learn the $i$-th entry $\mat D_i$ of $\mat D$ without revealing the index $i$. A client will send a query $\mat q$ to the server, who then responds with a response $\mat r$. A non-trivial PIR scheme requires \emph{compactness}, namely that the response length is less than the database length.

We now consider the following candidate PIR scheme from Dense-Sparse LPN, which borrows from our construction of LTDFs and the template in~\cite{FOCS:KusOst97}. Assume for simplicity that the database $\mat D \in \F_2^m$ is $t$-sparse (the general case can be reduced to this setting by sparsification). To make a query on index $i \in [m]$, the client will send an encryption
\begin{align*}
    \mat q = (\mat A, \mat b = \mat s \mat A + \mat e + \mat u_i), \quad\text{ where }\quad \mat A = \mat T \mat M \in \F_2^{\frac{n}{2} \times m} \quad \text{ and }\quad  \mat u_i = \underbrace{(0,\dots,1,\dots,0)}_{i\text{-th position}}.
\end{align*}
The server will respond with the ciphertext applied to $\mat D$, i.e. with $\mat r = (\mat r_1,\mat r_2) =  (\mat A \mat D, \mat b \mat D) \in \F_2^{\frac{n}{2} + 1}$. The client then recovers $\mat D_i = \langle \mat D, \mat u_i\rangle$ by computing $\mat r_2 - \mat s \mat r_1 = \mat D \mat u_i + \mat D \mat e$. By sampling $\mat e$ from a Bernoulli distribution of probability $\epsilon = O(1/t)$, we can guarantee that the scheme has constant correctness, which can be amplified to negligible by $\omega(\log n)$ repetitions of $\mat b$. Client's query privacy also follows directly from Dense-Sparse LPN with noise rate $\epsilon$.

The problem with the scheme above is that it is not compact. Indeed, the response is of length $\frac{n}{2}+1$ which is greater than the database length $\log {m \choose t} \approx t \log(m/t)$. Here, we are in a similar situation to our CRHF construction, but we \emph{cannot} use the same trick to achieve compactness here. The reason is that, if we were to multiply $\mat r_1 = \mat A \mat D$ by a further compressing matrix $\mat U \in \F_2^{\ell \times \frac{n}{2}}$, then the client would receive $\mat U \mat A \mat D$. From there, the client must ``decompress'' $\mat A \mat D$ from $\mat U$, but this appears computationally infeasible since $\mat U$ does not come with a trapdoor for efficient inversion!


We note that the same difficulty above also prevents us from using Dense-Sparse LPN to build laconic oblivious transfer and identity-based encryption~\cite{EC:BLSV18}. Therefore, we leave as open question the task of overcoming these limitations, either by finding a way to efficiently compress and decompress $\mat A \mat x$ for a random $t$-sparse $\mat x$, or by proposing a different LPN variant that avoids this issue entirely.

\section{Preliminaries}\label{sec:prelim}

\paragraph{Notation.} Let $\mathbb{N}=\{1,2,\dots\}$ be the natural numbers, and define $[a,b]:=\{a,a+1,\dots,b\}$, $[n]:=[1,n]$. Our logarithms are in base $2$. For a finite set $S$, we write $x \sample S$ to denote uniformly sampling $x$ from $S$. We denote the security parameter by $\secparam$; our parameters depend on $\secparam$, e.g. $n=n(\secparam)$, and we often drop the explicit dependence.

We abbreviate PPT for probabilistic polynomial-time. Our adversaries are non-uniform PPT, or equivalently, polynomial-sized, ensembles $\calA=\{\calA_\secparam\}_{\secparam \in \N}$. We write $\negl(\secparam)$ to denote a negligible function in $\secparam$. Two ensembles of distributions $\{\calD_\lambda\}_{\lambda \in \N}$ and $\{\calD'_\lambda\}_{\lambda \in \N}$ are computationally indistinguishable if for any non-uniform PPT adversary $\calA$ there exists a negligible function $\negl$ such that $\calA$ can distinguish between the two distributions with probability at most $\negl(\lambda)$.

For $q \in \N$ that is a prime power, we write $\F_q$ to denote the finite field with $q$ elements, and $\F_q^\times$ to denote its non-zero elements.
We write vectors and matrices in boldcase, e.g. $\mat v \in \F_q^m$ and $\mat A \in \F_q^{n \times m}$. Given $\mat v \in \F_q^m$, we define the Hamming weight $\wt(\mat v)$, also denoted $\norm{\mat v}_0$, to be the number of non-zero entries of $\mat v$.

\paragraph{Bernoulli Distribution.} We denote the Bernoulli distribution over a finite field $\F_q$ with noise rate $\epsilon \in (0,1)$ by $\Ber(\F_q,\epsilon)$; this distribution gives $0$ with probability $1 - \epsilon$, and a random non-zero element of $\F_q$ with probability $\epsilon$. We write $e \sim \Ber(\F_q,\epsilon)$ to denote that $e$ comes from the corresponding Bernoulli distribution. When $q=2$, we omit $\F_q$ and simply write $\Ber(\epsilon)$.

\begin{definition}[Bias]\label{def:bias}
    Let $\F$ be a finite field. Given a distribution $\dist$ over $\F^m$ and a vector $\mat u \in \F^m$, we define the \emph{bias} of $\dist$ with respect to $\mat u$ to be 
    \begin{align*}
        \bias_{\mat u}(\dist) := \left| \E_{\mat x \sim \calD}\left[\langle \mat u, \mat x\rangle\right] - \frac{1}{\abs{\F}}\right| .
    \end{align*}
    The \emph{bias} of $\dist$ is defined as $\bias(\dist) = \max_{\mat u \ne \mathbf{0}}\bias_{\mat u}(\dist)$.
\end{definition}

\begin{lemma}[Bias of the Bernoulli distribution]\label{lem:bias-bernoulli}
    For any finite field $\F_q$, noise rate $\epsilon \in (0,1)$, and $d \in \N$, consider the noise distribution $\dist_{m,n}=(\Ber(\F_q,\epsilon))^m$. Then $\epsilon_d = \left(1-\frac{q}{q-1}\epsilon\right)^d$.
\end{lemma}

As a special case when $q=2$, we have the piling-up lemma.

\begin{lemma}[Piling-Up Lemma]\label{lem:piling-up}
    For any $\epsilon \in (0,1)$, we have that 
    \[\Pr\left[ \sum_{i=1}^\ell e_i = 1  \;\middle|\; e_1,\dots,e_\ell \gets \Ber(\epsilon)\right] = \frac{1 - (1-2\epsilon)^\ell}{2} < \min\left(\epsilon \ell,\frac{1}{2} - 2^{-4\epsilon \ell-1}\right).\]
\end{lemma}

\paragraph{Tail Bounds.} We also state some standard tail bounds for binary variables.

\begin{lemma}[Chernoff/Hoeffding bound]\label{lem:chernoff-bound}
    Let $X_1,\dots,X_n \in \{0,1\}$ be i.i.d random variables with mean at most $\epsilon$. Then for every $\kappa > 1$,
    \[\Pr[X_1 + \dots + X_n > (1+\kappa)\epsilon n] \le e^{-2\kappa^2 \epsilon n}.\]
\end{lemma}

\paragraph{Binomial Approximation.} We will use the following basic approximations of the binomial coefficient, which can be found in e.g.~\cite{cormen2022introduction}. Namely, for any $1 \le k \le n/2$, we have
\begin{equation}\label{eqn:binomial-approx}
    \left(\frac{n}{k}\right)^k \le {n \choose k} \le \left(\frac{en}{k}\right)^k.
\end{equation}

\paragraph{Hamming Balls.} Given $n,w \in \N$ with $w \le n$, we define the following sets:
\begin{itemize}
    \item $\calB_{\le}(n,w) = \{ \mat x \in \{0,1\}^n \mid \wt(\mat x) \le w\}$,
    \item $\calB(n,w) = \{\mat x \in \{0,1\}^n \mid \wt(\mat x) = w\}$,
    \item $\calB_\reg(n,w) = \{\mat x = \mat x_1 \| \dots \| \mat x_{w} \in \{0,1\}^n \mid \mat x_i \in \{0,1\}^{n/w} \,\land\, \wt(\mat x_i) = 1 \; \forall\, i \in [w]\}$, for any $w$ dividing $n$.
\end{itemize}
Their sizes are as follows.
\begin{lemma}\label{lem:hamming-ball-sizes}
    For any $w < n/3$, we have the following:
    \begin{itemize}
        \item $\abs{\calB_{\le}(n,w)} = \sum_{t=0}^w {n \choose t} \in \left( {n \choose w}, 2{n \choose w}\right)$.
        \item $\abs{\calB(n,w)} = {n \choose w} \in \left( 2^{w \log (n/w)}, 2^{w \log(en/w)} \right)$.
        \item $\abs{\calB_\reg(n,w)} = \left(\frac{n}{w}\right)^{w} = 2^{w\log(n/w)}$.
    \end{itemize}
    Thus, we have that $\abs{\calB_\reg(n,w)} < \abs{\calB(n,w)} < \abs{\calB_{\le}(n,w)} < 2^{w \log e + 1} \cdot \abs{\calB_{\reg}(n,w)}$.
\end{lemma}

\begin{proof}
    The only non-trivial claim is that $\sum_{t=0}^w {n \choose t} < 2 {n \choose w}$. This follows from the fact that ${n \choose t-1} = \frac{t}{n-t+1} {n \choose t} < \frac{1}{2} {n \choose t}$ for all $t < n/3$, and hence $\sum_{t=0}^w {n \choose t} < \left(1+\frac{1}{2}+\left(\frac{1}{2}\right)^2+\dots \right) {n \choose w} = 2 {n \choose w}$.
\end{proof}

\paragraph{Sparsification and the Gadget Matrix.} We describe an LPN analogue to the LWE gadget matrix~\cite{EC:MicPei12}, based on the idea of sparsification.

Given any $s \in \N$, we denote by $\bin_s: [0,2^s-1] \to \{0,1\}^s$ the binary decomposition function, and $\bin_s^{-1}: \{0,1\}^s \to [0,2^s-1]$ its inverse. Given any $\mat v \in \{0,1\}^s$, we define $\ind(\mat v)$ to be the vector of length $2^s$ (indexed from $0$) consisting of all zeros, except for the $\bin_s^{-1}(\mat v)$-th entry being $1$.

\begin{definition}\label{def:gadget-matrix}
    Let $n,w \in \N$ with $w \le n$, and define $\tilde{w} = w \log(n/w)$. Given $\mat x \in \{0,1\}^{\tilde{w}}$ divided into $w$ blocks of size $s = \log(n/w)$, so that $\mat x = \mat x_1 \| \dots \| \mat x_w$, we define its \emph{$(n,w)$-sparsification} $\spfy_{n,w}(\mat x)$ to be \[\spfy_{n,w}(\mat x) := \ind(\mat x_1) \| \dots \| \ind(\mat x_{w}) \in \{0,1\}^n.\] 
We also define the \emph{$(n,w)$-gadget matrix} $\mat G_{n,w}$ to be 
\[\mat G_{n,w} := \mat g_s \otimes \mat I_w \in \{0,1\}^{w \times n},\quad\text{where }\mat g_s := \left[\bin_s(0) \| \bin_s(1) \| \dots \| \bin_s(2^s-1)\right],\]
and $\mat I_{w}$ is the identity matrix of dimension $w$. For convenience, we also refer to $\spfy_{n,w}$ as $\mat G_{n,w}^{-1}$, and omit the subscripts $n,w$ when they are clear from context. 
\end{definition}
We can easily check that 
\begin{align}\label{eqn:gadget-inversion}
    \mat G_{n,w}\cdot \mat G_{n,w}^{-1}(\mat x) = \mat G_{n,w} \cdot \spfy_{n,w}(\mat x) = \mat x.
\end{align}
Note that for the dimensions of $\mat G_{n,w}$ to satisfy $n = \poly(w)$, we need $s \le c \log w$ for some constant $c$. We also get the following identity for any $\mat x, \mat y \in \{0,1\}^{\tilde{w}}$:
\begin{align}\label{eqn:gadget-inner-product}
    \langle \mat x, \mat y \rangle = \left\langle \mat x, \mat G_{n,w} \cdot \spfy_{n,w}(\mat y) \right\rangle = \left\langle \mat G_{n,w}^\top \cdot \mat x, \spfy_{n,w}(\mat y) \right\rangle.
\end{align}


\subsection{Coding Theory}

\begin{definition}[Dual Distance]\label{def:dual-distance}
Let $\F_q$ be a finite field and $n < m \in \N$. The \emph{dual distance} of a matrix $\mat A \in \F_q^{n \times m}$, denoted $\dd(\mat A)$, is defined to be minimum sparsity of a vector $\mat x \in \F_q^m$ in the kernel of $\mat A$. In other words, we define $\dd(\mat A) = \min\{\wt(\mat x) \mid \mat A \mat x = \mat 0\}$.
\end{definition}

\begin{definition}[$q$-ary Entropy]
    For any $x \in (0,1)$, we define the $q$-ary entropy function to be $H_q(x) = x\log_q(q-1) - x\log_q x - (1-x)\log_q (1-x)$. We denote binary entropy by $H(x)$.
\end{definition}

The following standard results can be found in coding theory textbooks, e.g.~\cite{web-coding-book}.

\begin{lemma}[Entropy Approximation]\label{lem:entropy-approx}
    For any $\delta \in (0,1)$, we have that $\delta \log(1/\delta) < H(\delta) < \delta \log(4/\delta)$.
\end{lemma}

\begin{lemma}[Gilbert-Varshamov Bound]\label{lem:G-V}
    Let $\F_q$ be a finite field. For every $\delta \in (0,1-1/q)$ and every $\epsilon \in (0,1-H_q(\delta))$, letting $k=\lfloor (1-H_q(\delta)-\epsilon)\cdot n \rfloor$, a random matrix $\mat G \gets \F_q^{k \times n}$ generates a $q$-ary linear code of distance at least $\delta n$ with probability at least $1-q^{-\epsilon n}$.

    Equivalently, for $\ell = \lceil (H_q(\delta) + \epsilon) \cdot n \rceil$, a random matrix $\mat H \gets \F_q^{\ell \times n}$ is a parity-check matrix for a $q$-ary linear code of distance at least $\delta n$ with probability at least $1-q^{-\epsilon n}$.
\end{lemma}

\begin{lemma}[Asymptotically Good Codes]\label{lem:asymp-good-codes}
    For some constants $\delta,\rho > 0$, there exists an explicit construction of a family of binary codes $\{\mat C_n\}_{n \in \N}$ with block length $N(n)$ (bounded by some fixed polynomial in $n$), rate $\rho$, and supporting efficient error correction for up to $\delta n$ errors.
\end{lemma}

\section{Our Code-Based Assumption: Dense-Sparse LPN}\label{sec:assumption}

In this section, we define our main assumption, Dense-Sparse LPN. To do so, we will first define a general LPN assumption with an arbitrary distribution of coefficient matrix.

\begin{definition}[Decisional $(\calM,\epsilon)$-LPN]\label{def:lpn}
    Let $n \in \N$ be the dimension, $m=m(n)$ be the number of samples, $\epsilon=\epsilon(n) \in (0,1)$ be the noise rate, and $q=q(n)$ be a prime power. Given an efficiently sampleable distribution $\calM = \calM(n,m,\F_q)$ over matrices in $\F_q^{n \times m}$, we say that the \emph{$(\calM,\epsilon)$-LPN assumption} is $(T(n),\delta(n))$-hard if for all adversary $\calA$ running in time at most $T$, the following holds:
    \begin{align*}
        \Adv^{\LPN}_{n,m,q,\calM,\epsilon}(\calA) := \left|\begin{aligned}
            &\Pr\left[\calA(\mat A, \mat s \mat A + \mat e) = 1 \;\middle|\; \mat A \gets \calM, \mat s \gets \F_q^{1 \times n}, \mat e \gets \Ber(\F_q,\epsilon)^{1 \times m}\right] \\
            - \; &\Pr\left[\calA(\mat A, \mat u) = 1 \;\middle|\; \mat A \gets \calM, \mat u \gets \F_q^{1 \times m}\right]
        \end{aligned}\right| \le \delta.
    \end{align*}
    We say that $(\calM,\epsilon)\text{-}\LPN$ is \emph{polynomially hard} if for every polynomial $p(n)$, there exists a negligible function $\negl(n)$ such that it is $(p(n),\negl(n))$-hard. Similarly, $(\calM,\epsilon)\text{-}\LPN$ is \emph{subexponentially hard} if there exists a constant $0 < c < 1$ and a negligible function $\negl(n)$ such that it is $(2^{n^c},\negl(n))$-hard.
\end{definition}

\begin{definition}[Decisional $(\calM,\epsilon)$-Dual LPN]\label{def:dual-lpn}
    Consider $n,m,q,\calM,\epsilon$ as in \Cref{def:lpn}. We say that the \emph{$(\calM,\epsilon)$-dual LPN assumption} is $(T(n),\delta(n))$-hard if for all adversary $\calA$ running in time at most $T$, the following holds:
    \begin{align*}
        \Adv^{\mathsf{dual}\text{-}\LPN}_{n,m,q,\calM,\epsilon}(\calA) := \left|\begin{aligned}
            &\Pr\left[\calA(\mat H, \mat H \mat e) = 1 \;\middle|\; \mat H \gets \calM, \mat e \gets \Ber(\F_q,\epsilon)^{1 \times m}\right] \\
            - \; &\Pr\left[\calA(\mat H, \mat u) = 1 \;\middle|\; \mat H \gets \calM, \mat u \gets \F_q^{1 \times m}\right]
            \end{aligned}\right| \le \delta.
    \end{align*}
    We may define polynomial or subexponential hardness of $(\calM,\epsilon)\text{-}\mathsf{dual}\text{-}\LPN$ similar to~\Cref{def:lpn}.
\end{definition}

\begin{remark}
    For the rest of the paper, we will work over the binary field (so that $q=2$). We note that our assumptions and constructions can be straightforwardly generalized to work with any constant $q=O(1)$.
\end{remark}

When $\calM$ is the uniform distribution over $\F_2^{n \times m}$, we recover the (standard) LPN assumption~\cite{C:BFKL93} (with Dual-LPN equivalent to LPN). We now define variants of LPN with different matrix distributions $\calM$.

\begin{definition}[Sparse LPN]\label{def:sparse-lpn}
    Let $k \in \N$ be a constant, and consider parameters $n \in \N$, $m=m(n) < n^{k/2}$, and $\epsilon=\epsilon(n) < 1$. Denote by $\SpMat(n,m,k)$ the set of matrices $\mat A \in \F_2^{n \times m}$ such that each column of $\mat A$ has exactly $k$ non-zero entries.

    We define the \emph{$(n,m,k,\epsilon)$-sparse LPN ($\sLPN$) assumption} to be the following: there exists an efficiently sampleable distribution $\calM_{\sparse}$ over $\SpMat(n,m,k)$ such that the \emph{$(\calM_{\sparse},\epsilon)\text{-}\LPN$} assumption holds.
\end{definition}

Note that the above assumption does not specify the exact distribution $\calM_{\sparse}$ of sparse matrices. This is because the ``canonical'' distribution $\calM_{\unif}$ of sampling $k$-sparse columns independently and uniformly at random do \emph{not} suffice for the above assumption. The reason is that there is a noticeable probability, of roughly $O(m^2/n^k)$, for sampling a ``bad'' matrix $\mat A \gets \calM_{\unif}$ with small dual distance, which would break the assumption by giving the adversary a noticeable distinguishing advantage.

Motivated by this discussion, we lay out a necessary criteria for the sparse matrix distribution $\calM_{\sparse}$ to satisfy the Sparse LPN assumption.

\begin{definition}[Sparse matrices with $\omega(1)$-dual distance]\label{def:good-dist}
    For every $n \in \N$, $k=k(n)$, $m=m(n) < n^{k/2}$ and $d = \omega(1)$, define $\SpMat(n,m,k,d) = \{\mat A \in \SpMat(n,m,k) \mid \dd(\mat A) \ge d\}$ to be the subset of $\SpMat(n,m,k)$ consisting of matrices with super-constant dual distance of at least $d$.

    We say that an efficiently sampleable distribution $\calM_{\sparse}$ over $\SpMat(n,m,k)$ is \emph{$(d,\delta)$-good} if \[\Pr[\mat A \not\in \SpMat(n,m,k,d) \mid \mat A \gets \calM_{\sparse}] \le \delta.\] We say that $\calM_{\sparse}$ is \emph{good} if it is $(d,\delta)$-good for some $d=\omega(1)$ and $\delta=\negl(n)$.
\end{definition}

We may make the \emph{conjecture} that any such good distribution $\calM_{\sparse}$ would give rise to a secure $\SLPN$ assumption. A more conservative assumption would be to assume $\SLPN$ assumption holds for \emph{specific} good distributions that have been constructed in the literature. In our work, we will use the following distribution by Applebaum and Kachlon~\cite{FOCS:AppKac19}.

\begin{theorem}[Theorem 7.18~\cite{FOCS:AppKac19}, adapted]\label{thm:ak19-dist}
    For every even $k \ge 6$, every $1 < c < k/4$ with $\gamma = k - 4c$, there exists an efficiently computable, $\left(O(n^{\delta}),n^{-O\left(\frac{\log \log \log n}{\log \log \log \log n}\right)}\right)$-\hyperref[def:good-dist]{good} distribution over $\SpMat(n,m,k)$, where $m = n^c$ and $\delta = \frac{k - 4c - \gamma}{k - \gamma - 4}$. We call this the \emph{AK19 distribution}.
\end{theorem}

Note that the AK19 distribution does not give us all possible parameter range for $k$-sparse matrices. In particular, the number of samples $m$ is limited to be at most $n^{k/4}$. However, the parameter regime that the AK19 distribution supports overlap with the compression regime of \Cref{lem:compress} for any constant compression factor $D > 1$. Therefore, we can indeed use this distribution to instantiate our schemes and assumptions.


\begin{definition}[Dense-Sparse LPN]\label{def:dense-sparse-lpn}
    Let $k \in \N$, $\alpha \in (0,1)$ be constants, and consider parameters $n \in \N$, $m=m(n) < n^{k/2}$, and $\epsilon=\epsilon(n) < 1$. Let $\calM_{\sparse}$ be a \hyperref[def:good-dist]{good} distribution over $\SpMat(n,m,k)$.
    We define the \emph{$(n,m,k,\calM_{\sparse},\epsilon)$-Dense-Sparse LPN ($\DSLPN$)} assumption to be the $(\calM,\epsilon)$-LPN assumption, where $\calM$ is the following distribution:
    \begin{align*}
        \calM = \{\mat T \cdot \mat M \mid \mat T \gets \F_2^{\alpha n \times n}, \mat M \gets \calM_{\sparse}\}.
    \end{align*}
    In other words, we say that Dense-Sparse LPN is $(T(n),\delta(n))$-hard if the following holds for every adversary $\calA$ running in time at most $T$:
    \begin{align*}
        \Adv^{\DSLPN}_{n,m,k,\calM_\sparse,\epsilon}(\calA) := \left| \Pr\left[\calA(\mat A, \mat s \mat A + \mat e) = 1\right] - \Pr\left[\calA(\mat A, \mat u) = 1\right] \right| \le \delta,
    \end{align*}
    where $\mat T \gets \F_2^{\alpha n \times n}$, $\mat M\gets \calM_{\sparse}$, $\mat A = \mat T \cdot \mat M$, $\mat s \gets \F_2^{1 \times \alpha n}$, $\mat e \gets \Ber(\epsilon)^{1 \times m}$, and $\mat u \gets \F_2^m$.
\end{definition}

To simplify notation, we will take $\alpha=1/2$ for all of our constructions, though the assumption plausibly holds for any constant $\alpha \in (0,1)$. We provide detailed cryptanalysis of Dense-Sparse LPN in \Cref{sec:cryptanalysis}. 

\begin{remark}[Dense-Sparse Dual-LPN]\label{rem:dense-sparse-dual-lpn}
    It is plausible that the dual version of Dense-Sparse LPN also holds, namely that \begin{align*}
        \{(\mat T \mat M, \mat T \mat M \mat e)\} \approx_c \{(\mat T \mat M, \mat u)\} \quad \text{ for }\quad\text{$\mat T \gets \F_2^{\alpha n \times n}$, $\mat M \gets \calM_{\sparse}$, $\mat e \gets \Ber(\epsilon)^{m \times 1}$, and $\mat u \gets \F_2^{\alpha n \times 1}$}.
    \end{align*}
    This is because $\mat M \mat e$ is roughly $O(kt)$-sparse, so that plausibly $(\mat T,\mat T \mat M \mat e) \approx_c (\mat T,\mat u)$ by dual-LPN for the random matrix $\mat T \in \F_2^{\alpha n \times n}$.
    
    Note that this is not yet a rigorous argument since the sparse vector $\mat M \mat e$ is very far from being Bernoulli distributed.
    Since we do not use Dense-Sparse dual-LPN in our constructions, we leave a more rigorous cryptanalysis and further applications of this assumption to future work.
\end{remark}

\paragraph{Compression Regime.} We give here the parameter regime that allows for the function \[f_{\mat M}: \calB_{\reg}(m,t) \to \Ball_{\le}(n,kt)\quad \text{ defined by } \quad f_{\mat M}(\mat x) = \mat M \cdot \mat x,\quad\text{ for any }\quad \mat M \in \SpMat(n,m,k),\] to be compressing (by some constant factor $D$ in the exponent).

\begin{lemma}\label{lem:compress}
    Given constant $k \in \N, k \ge 3$ and compression factor $D > 1$, define $\delta_{(\ref{lem:compress})}(k,D) := 1-\frac{k/2-1}{Dk-1}$. For any $\delta \in (\delta_{(\ref{lem:compress})}(k,D),1)$, any $n \in \N$ large enough, and any $m<n^{k/2}$, letting $t=n^\delta$, we have the following:
    \begin{align}\label{eqn:compression-ratio}
        \abs{\Ball_\reg(m,t)} = 2^{t \log(m/t)} > \left(2^{kt\log (en/kt) + 1}\right)^D > \abs{\Ball_{\le}(n,kt)}^D
    \end{align}
    whenever $m > m_{(\ref{lem:compress})}(n,k,D,\delta) := n^{1+(Dk-1)(1-\delta)}$.
\end{lemma}

\begin{proof}
    Note that the last inequality in \Cref{eqn:compression-ratio} follows from \Cref{lem:hamming-ball-sizes}. Thus it suffices to show the central inequality.
    Taking the logarithms of both sides, we need to show that
    \begin{align*}
        t (\log m - \log t) > D\left(kt \log(en/kt) + 1\right).
    \end{align*}
    Dividing by $t$ and isolating out $\log(m)$, we get
    \begin{align*}
        \log m > D/t + \log(t) + Dk \log(en/kt),
    \end{align*}
    which is equivalent to
    \begin{align*}
        m > 2^{D/t} \cdot t \cdot \left(\frac{en}{kt}\right)^{Dk} = 2^{D/t} \cdot \left(\frac{e}{k}\right)^{Dk} \cdot \frac{n^{Dk}}{t^{Dk-1}}.
    \end{align*}
    Plugging in $t=n^\delta$ gives us
    \begin{align*}
        m > 2^{D/n^\delta} \cdot \left(\frac{e}{k}\right)^{Dk} \cdot n^{1+(Dk-1)(1-\delta)}.
    \end{align*}
    Since $\lim_{n \to \infty}2^{D/n^\delta} = 1$ and $e/k \le e/3 < 1$, for $n$ large enough we have \[2^{D/n^\delta} \cdot \left(\frac{e}{k}\right)^{Dk} < \left(\frac{e}{k}\right)^{-Dk} \cdot \left(\frac{e}{k}\right)^{Dk} = 1.\]
    Therefore, \Cref{eqn:compression-ratio} is satisfied when $m > n^{1+(Dk-1)(1-\delta)}$.
    Finally, since we need to impose the condition that $m < n^{k/2}$, it follows that $\delta$ is at most $1-\frac{k/2-1}{Dk-1} = \delta_{(\ref{lem:compress})}(k,D)$.
\end{proof}
\section{Collision-Resistant Hash Functions}\label{sec:CRHF}

In this section, we give a simple construction of collision-resistant hash functions (CRHFs) from our Dense-Sparse LPN assumption. While lossy trapdoor functions (LTDFs), which we construct in \Cref{sec:ltdf}, are known to imply CRHFs~\cite{STOC:PeiWat08}, the resulting construction is more complex than the one presented here (and require a slightly smaller noise rate $\epsilon = O(1/n^{\delta})$ instead of $\epsilon = O(\log n / n^{\delta})$).

Compared to prior LPN-based hashes~\cite{EC:BLVW19,AC:YZWGL19}, which rely on either inverse quasi-polynomial noise rate $\epsilon=O(\log^2 n / n)$ or sub-exponential hardness, our construction only requires polynomial hardness and sub-exponential noise rate $\epsilon = \widetilde{O}(n^{-\delta})$ for $\DSLPN$, where $\delta$ is a constant that can be arbitrarily close to $3/4$ (as $k \to \infty$).

\begin{definition}\label{def:CRHF}
    A \emph{collision-resistant hash function (CRHF)} family, with input length $m(\lambda)$ and output length $n(\lambda)$, is a tuple of PPT algorithms $\calH = (\Gen,\Hash)$ with the following properties:
 \begin{itemize}
     \item \textbf{Syntax:}
     \begin{itemize}
         \item $\Gen(1^\lambda) \to \sfk$. On input the security parameter $1^\lambda$, output a hash key $\sfk$.
         \item $\Hash(\sfk,\mat x)$. On input the hash key $\sfk$ and an input $\mat x \in \{0,1\}^{m(\lambda)}$, deterministically output $\sfh \in \{0,1\}^{n(\lambda)}$.
     \end{itemize}

     \item \textbf{Compression:} We have $m(\lambda) > n(\lambda) + 2 \lambda$ for all $\lambda \in \N$.

     \item \textbf{Collision-Resistance:} For all polynomial-size adversary $\calA$, the following probability is negligible:
     \begin{align*}
         \Adv^{\CRHF}(\calA) := \Pr\left[\begin{gathered}
            \mat x_1 \ne \mat x_2 \quad\land\\
            \Hash(\sfk,\mat x_1) = \Hash(\sfk,\mat x_2)
        \end{gathered} \; \middle| \; \begin{aligned}
            &\sfk \gets \Gen(1^\lambda) \\
            &(\mat x_1, \mat x_2) \gets \calA(\sfk)
        \end{aligned}\right] \le \negl(\lambda).
     \end{align*}
 \end{itemize}
\end{definition}

\begin{figure}
    \centering
    \begin{construction}[Collision-Resistant Hash Function from Dense-Sparse LPN]
        \textbf{Parameters.}
        \begin{itemize}
            \item Pick any constant $k \in \N, k \ge 3$, any constant $D > 2$, and any $\delta \in (\delta_{(\ref{lem:compress})}(k,D),1)$. Let $\rho = 1/2-1/D > 0$.
            \item Let $m > m_{(\ref{lem:compress})}(n,k,D,\delta)$, $t = n^\delta$, and $\epsilon = O(\log n / t)$. Denote $\tilde{t} = t \log(m/t)$. Choose $n = (2\lambda/\rho)^{1/\delta}$ so that $\rho\tilde{t} > \rho n^{\delta} > 2\lambda$. Let $\calM_{\sparse}$ be a \hyperref[def:good-dist]{good} distribution over $\SpMat(n,m,k)$.
        \end{itemize}

        \textbf{Construction.}
        \begin{itemize}
            \item $\Gen(1^\lambda) \to \sfk$. Sample $\mat H \gets \F_2^{(1-\rho)\tilde{t}\times n}$, and $\mat M \gets \calM_{\sparse}$. Return $\sfk=\mat A' = \mat H \cdot \mat M$.

            \item $\Hash(\sfk,\mat x) \to \sfh$. On input $\mat x \in \F_2^{\tilde{t}}$, return $\sfh = \mat A' \cdot \spfy_{m,t}(\mat x) \in \F_2^{(1-\rho)\tilde{t}}$.
        \end{itemize}
    \end{construction}
    \caption{CRHF construction}
    \label{fig:CRHF}
\end{figure}

\begin{theorem}\label{thm:CRHF-secure}
    Assuming the $(n,m,k,\calM_{\sparse},\epsilon)$-$\DSLPN$ assumption holds, where $n,m,k,\calM_{\sparse},\epsilon$ are defined as in \Cref{fig:CRHF}. Then \Cref{fig:CRHF} gives a construction of a CRHF family.
\end{theorem}

\begin{proof}
    It is clear that we have compression, as the gap between input and output size is $\tilde{t} - (1-\rho)\tilde{t} = \rho\tilde{t} > 2\lambda$ by our parameter choice. It remains to show collision-resistance, which we show by going through the following hybrids.
    
    \begin{itemize}
        \item $\Hyb_0$: this is the CHRF game, where an adversary $\calA$ receives $\mat A' = \mat H \cdot \mat M$ and outputs $\mat x_1, \mat x_2 \in \F_2^{t\log(m/t)}$. $\calA$ wins, or equivalently $\Hyb_0$ returns $1$, if $\mat x_1 \ne \mat x_2$ and $\mat A' \cdot \spfy_{m,t}(\mat x_1) = \mat A' \cdot \spfy_{m,t}(\mat x_2)$; equivalently, when $\mat A' \cdot \mat x' = \mat 0$ where $\mat x' = \spfy_{m,t}(\mat x_1) - \spfy_{m,t}(\mat x_2)$.
        \item $\Hyb_1$: this is identical to $\Hyb_0$, except the adversary $\calA$ only wins if $\mat M \cdot \mat x' = \mat 0$ (and also $\mat x_1 \ne \mat x_2$). Thus, $\Hyb_2$ differs from $\Hyb_1$ only when $\mat M \cdot \mat x' \ne \mat 0$ but $\mat H \cdot (\mat M \cdot \mat x') = \mat 0$. Since $\mat y = \mat M \cdot \mat x' \in \F_2^n$ is at most $2kt$-sparse, this amounts to bounding the probability that a random parity-check matrix $\mat H \in \F_2^{(1-\rho)\tilde{t} \times n}$ has distance less than $2kt$.

        We will show that this probability is negligible by examining our parameter choices. By \Cref{lem:entropy-approx}, we have that
        \begin{align*}
            H\left(\frac{2kt}{n}\right)\cdot n &< 2kt \log\left(\frac{2n}{kt}\right).
        \end{align*}
        Since we choose $t=n^\delta$ and $m > m_{(\ref{lem:compress})}(n,k,D,\delta)$, we know from \Cref{lem:compress} that the following holds:
        \begin{align*}
            &\qquad\quad \tilde{t} = t \log\left(\frac{m}{t}\right) > D \left(kt \log\left(\frac{en}{kt}\right) + 1\right) > D kt \log\left(\frac{2n}{kt}\right) \\
            &\implies \frac{2}{D}\,\tilde{t} > 2 kt \log\left(\frac{2n}{kt}\right) > H\left(\frac{2kt}{n}\right)\, n\\
            &\implies (1-\rho)\tilde{t} = \left(\frac{2}{D} + \rho\right) \tilde{t} > H\left(\frac{2kt}{n}\right) n + 2\lambda.
        \end{align*}
        Therefore, by Gilbert-Varshamov (\Cref{lem:G-V}), $\mat H \in \F_2^{(1-\rho)\tilde{t} \times n}$ has distance less than $2kt$ with probability at most $2^{-2\lambda} = \negl(\lambda)$.
        
        \item $\Hyb_2$: this is identical to $\Hyb_1$, but instead of sampling $\mat H \gets \F_2^{2t\log(m/t) \times n}$, we sample $\mat H = \mat H' \cdot \mat T$ where $\mat H' \gets \F_2^{2t\log(m/t) \times n/2}$ and $\mat T \gets \F_2^{n/2 \times n}$. Since $\widetilde{t} < n/2 < n$, by basic linear algebra, it follows that $\mat H$ is identically distributed as in $\Hyb_1$.
        
        We now show that $\Pr[\Hyb_2 \text{ returns } 1] \le \negl(\lambda)$ assuming $\DSLPN$ holds. We do this by constructing an adversary $\calB$ against $\DSLPN$ from an adversary $\calA$ against $\Hyb_2$. $\calB$ receives $(\mat A, \mat b)$ where $\mat A = \mat T \cdot \mat M$ for $\mat T \gets \F_2^{n/2 \times n}$, $\mat M \gets \calM_{\sparse}$, and $\mat b \in \F_2^{1 \times m}$ is either uniformly random or is equal to $\mat s \mat A + \mat e$ for a random $\mat s \gets \F_2^{1 \times n/2}$ and $\mat e \gets \Ber(\epsilon)^{1 \times m}$. $\calB$ now samples $\mat H' \gets \F_2^{(1-\rho)\tilde{t} \times n/2}$, computes $\mat A' = \mat H' \cdot \mat A$, then runs $\calA$ on $\mat A'$ to get $(\mat x_1,\mat x_2)$ from $\calA$ and return $\ans:=\mat b \cdot (\spfy_{m,t}(\mat x_1) - \spfy_{m,t}(\mat x_2))$. Here $\ans=1$ indicates that $\calB$ received random $\mat b \gets \F_2^{1 \times m}$.

        We will analyze $\calB$'s success probability. If $\mat b$ is uniformly random, then $\ans$ is also uniformly random whenever $\calA$ wins (so that $\mat x_1 \ne \mat x_2$); thus, in this case we have \[\Pr[\ans = 1 \mid \mat b \gets \F_2^{m}] = \frac{1}{2}\cdot \Pr[\Hyb_2 \text{ returns } 1].\]
        If $\mat b = \mat s \mat A + \mat e$, then $\ans = (\mat s \mat A + \mat e) \cdot \mat x' = \mat e \cdot \mat x'$ whenever $\calA$ wins in $\Hyb_2$. Since $\mat x'$ is at most $2t$-sparse, by \Cref{lem:piling-up} we have 
        \begin{align*}
             \Pr[\ans = 1 \mid \mat b = \mat s \mat A + \mat e] &\le \Pr[\mat e \cdot \mat x' = 1] \cdot \Pr[\Hyb_2 \text{ returns } 1]\\
             &\le \left(\frac{1 - (1-2\epsilon)^{2t}}{2}\right) \cdot \Pr[\Hyb_2 \text{ returns } 1]\\
             &\le \left(\frac{1}{2} - 2^{-8\epsilon t - 1}\right) \cdot \Pr[\Hyb_2 \text{ returns } 1]\\
             &= \left(\frac{1}{2} - \frac{1}{\poly(\lambda)}\right) \cdot \Pr[\Hyb_2 \text{ returns } 1].
        \end{align*}
        The last inequality is due to our choice of parameter $\epsilon = O(\log n/t)$, which implies $2^{-8\epsilon t - 1} = 1/\poly(n) = 1 / \poly(\lambda)$. Putting everything together, we have
        \begin{align*}
            \Adv^{\DSLPN}(\calB) &= \abs{\Pr[\ans = 1 \mid \mat b = \mat s \mat A + \mat e] - \Pr[\ans = 1 \mid \mat b \gets \F_2^m]} \\
            &\ge \frac{1}{\poly(\lambda)} \cdot \Pr[\Hyb_2 \text{ returns } 1].
        \end{align*}
        Therefore, if $\DSLPN$ holds (for our parameters), then \[\Pr[\Hyb_2 \text{ returns } 1] \le \poly(\lambda) \cdot \Adv^{\DSLPN}(\calB) \le \negl(\lambda).\]
    \end{itemize}
\end{proof}
\section{Lossy Trapdoor Functions}\label{sec:ltdf}

\subsection{Definition}\label{sec:ltdf-defn}

We define all-injective-but-one lossy trapdoor functions, which include lossy trapdoor functions as a special case of having two branches.

\begin{definition}\label{def:ltdf}
    An \emph{all-(injective)-but-one lossy trapdoor function} $\ABOLTDF$ with input size $n=n(\lambda)$, output size $m=m(\lambda)$, residual leakage $r=r(\lambda)$, and number of branches $B(\lambda)$, consists of a PPT algorithm $\Gen$ and a tuple of deterministic functions $(F,F^{-1})$ with the following syntax:
    \begin{itemize}
        \item $\Gen(1^\lambda,b^*) \to (\fk,\td)$. Given the security parameter $1^\lambda$ and a distinguished lossy branch $b^* \in [B(\lambda)]$, return a function key $\fk$ and a trapdoor $\td$.
        \item $F(\fk,b,x) \to y$. Given the function key $\fk$, a branch $b \in [B]$, and an input $x \in \F_2^{n}$, return an output $y \in \F_2^{m}$.
        \item $F^{-1}(\td,b,y) \to x$. Given the trapdoor $\td$, a branch $b \in [B]$, and an output $y \in \F_2^{m}$, return a preimage $x \in \F_2^{n}$.
    \end{itemize}
    and the following requirements:
    \begin{itemize}
        \item \textbf{Branch Indistinguishability.} For any two different branches $b^*_0 \ne b^*_1 \in [B]$, the following two distributions are indistinguishable
        \begin{align*}
            \left\{\fk \;\middle|\; (\fk,\td) \gets \Gen(1^\lambda,b^*_0)\right\}_{\lambda \in \N} \approx_c 
            \left\{\fk \;\middle|\; (\fk,\td) \gets \Gen(1^\lambda,b^*_1)\right\}_{\lambda \in \N}.
        \end{align*}
        \item \textbf{Correct Inversion for Injective Branch.} For any $\lambda \in \N$ and $b \ne b^* \in [B]$, we have
        \begin{align*}
            \Pr\left[ F^{-1}(\td,b,F(\fk,b,x)) = x \;\forall x \in \F_2^n \; \middle| \; (\fk,\td) \gets \Gen(1^\lambda,b^*) \right] \ge 1 - \negl(\lambda).
        \end{align*}
        \item \textbf{Lossyness for Lossy Branch.} For any $\lambda \in \N$, we have
        \begin{align*}
            \Pr\left[ \left\lvert \{F(\fk,b^*,x) \mid x \in \F_2^{n}\} \right\rvert \le 2^{r} \; \middle| \; (\fk,\td) \gets \Gen(1^\lambda,b^*) \right] \ge 1 - \negl(\lambda).
        \end{align*}
        We define the \emph{lossiness factor} to be $\Gamma = n/r$, meaning that in lossy mode, the output size is reduced by a factor of $\Gamma$ compared to the input size.
    \end{itemize}    
    A \emph{lossy trapdoor function} $\LTDF$ is an $\ABOLTDF$ with only two branches, i.e. $B(\lambda)=2$ for all $\lambda \in \N$.
\end{definition}

\paragraph{Other extension of LTDFs.} Over the years, several variants of LTDFs have been proposed with a goal toward increased functionality and more diverse applications. These include all-but-$N$ LTDFs~\cite{AC:HLOV11} (which can be built in a black-box way from LTDFs), all-but-many~\cite{EC:Hofheinz12}, and cumulative-all-lossy-but-one~\cite{PKC:ChaPraWic20} lossy trapdoor functions.
However, post-quantum constructions of the latter two variants~\cite{C:BoyLi17,C:LSSS17,EPRINT:LibNguPas22} rely on lattice techniques such as preimage sampling~\cite{STOC:GenPeiVai08,EC:MicPei12} and GSW-style homomorphic evaluation~\cite{C:GenSahWat13}. Since we do not know analogues of these techniques for code-based cryptography, we leave as future work the task of constructing these more advanced variants of LTDFs from code-based assumptions.

\subsection{Construction}\label{sec:ltdf-construction}

\begin{figure}[!ht]
    \centering
    \begin{construction}[ABO-LTDFs from Dense-Sparse LPN]

    \textbf{Parameters.} Let $k \in \N, k \ge 3$ be a constant, $\Gamma > 1$ be any desired lossiness factor.
    \begin{itemize}
        \item Let $D > \Gamma$ be the compression factor, and consider any $\delta \in (\delta_{(\ref{lem:compress})}(k,D),1)$.
        
        \item Let $n=\poly(\lambda)$, $m = m_{(\ref{lem:compress})}(n,k,D,\delta)$, $t=n^{\delta}$, and consider a \hyperref[def:good-dist]{good} distribution $\calM_{\sparse}$ over $\SpMat(n,m,k)$.

        \item Let $\calC=\{\mat C_{\kappa}\}_{\kappa \in \N}$ be an explicit family of asymptotically-good linear codes, where each $\mat C_{\kappa}$ has block length $L(\kappa)$, constant rate $\rho_{\calC}$, and admits an efficient algorithm $\Decode$ that can correct up to $\delta_{\calC} L$ errors (which exists by \Cref{lem:asymp-good-codes}).
        
        \item Let $D' = \frac{\Gamma D}{\Gamma - D}$ so that $\frac{1}{D} + \frac{1}{D'} = \frac{1}{\Gamma}$, and $\gamma = \min\left(\delta_{\calC},H^{-1}\left(\frac{\rho}{D'}\right)\right)$. Let $\alpha > 0$ be a constant such that $\frac{\alpha^2}{\alpha + 1} > \frac{\rho_{\calC}}{\gamma}$.
        
        \item Let $\ell = \frac{1}{\rho_{\calC}}\cdot t\log(\frac{m}{t})$ be the block length,\footnote{Without loss of generality, we may assume $\ell = L(\kappa)$ for some $\kappa \in \N$. Otherwise, letting $L(\kappa)$ be the nearest block size larger than $\ell$, we may scale $n$ (and thus $\ell$) by an appropriate polynomial amount so that $\ell = L(\kappa)$.} and $\epsilon = \frac{\gamma}{(\alpha+1)t}$ be the noise rate. We will abbreviate $\mat C = \mat C_{\ell} \in \F_2^{t\log(m/t) \times \ell}$.
    \end{itemize}

    \textbf{Construction.}
    \begin{itemize}
        \item $\Gen(1^\lambda,\tau^*) \to (\fk,\td)$. Given lossy branch $\tau^* \in \F_{2^{t\log(m/t)}}$, sample $\mat M \gets \calM_{\sparse}$, $\mat T \gets \F_2^{n/2 \times n}$, $\mat S \gets \F_2^{\ell \times n/2}$, and $\mat E \gets \Ber(\F_2,\epsilon)^{\ell \times m}$.
        
        Let $\mat A = \mat T \cdot \mat M \in \F_2^{n/2 \times m}$, and $\mat B = \mat S \cdot \mat A + \mat E + \mat C^\top \cdot \mat H_{\tau^*} \cdot \mat G_{m,t}^\top \in \F_2^{\ell \times m}$.

        Return $\fk = (\mat A, \mat B)$ and $\td = (\mat S,\tau^*)$.

        \item $F(\fk, \tau, \mat x) \to \mat y$. Given $\tau \in \F_{2^{t\log(m/t)}}$ and $\mat x \in \F_2^{t \log(m/t)}$, compute $\widetilde{\mat x} = \spfy_{m,t}(\mat x)$ and $\mat B_\tau = \mat B - \mat C^\top \cdot \mat H_{\tau} \cdot \mat G_{m,t}^\top$. 
        
        Return $\mat y = (\mat A \cdot \widetilde{\mat x}, \mat B_\tau \cdot \widetilde{\mat x}) \in \F_2^{n/2+\ell}$.

        \item $F^{-1}(\td,\tau,\mat y) \to \mat x$. Parse $\mat y = (\mat y_1 \in \F_2^{n/2}, \mat y_2 \in \F_2^\ell)$ and compute $\mat y' = \mat y_2 - \mat S \cdot \mat y_1$.
        
        Return $\mat x \gets \left(\mat H_{\tau^* - \tau}\right)^{-1} \cdot \Decode(\mat y')$.
    \end{itemize}
\end{construction}
    \caption{ABO-LTDF construction}
    \label{fig:LTDF}
\end{figure}

We give our construction of $\ABOLTDF$ in \Cref{fig:LTDF}. In the construction, we use the following family of matrices $\{\mat H_\tau\}_{\tau \in \F_{2^{n}}}$ for any $n \in \N$, where $\mat H_\tau \in \F_2^{n \times n}$ is the matrix corresponding to multiplication by $\tau$ in the field $\F_{2^n}$ (which is isomorphic as vector spaces to $\F_2^n$). It follows that $\mat H_{\tau} - \mat H_{\tau'} = \mat H_{\tau-\tau'}$ is invertible for all $\tau \ne \tau'$. Such a family has been used in previous works~\cite{EC:AgrBonBoy10,PKC:KilMasPie14}.

To analyze the security of our construction, we begin with the following result which bounds the noise growth.

\begin{lemma}\label{lem:ltdf-noise-bound}
    Consider the parameters $n,m,t,B,\epsilon,\ell$ as in \Cref{fig:LTDF}. Let $\mat E \gets \Ber(\F_2,\epsilon)^{\ell \times m}$. Then except with $\negl(\lambda)$ probability over the choice of $\mat E$, the vector $\mat E \cdot \spfy_{m,t}(\mat x) \in \F_2^{\ell}$ is at most $\gamma \ell$-sparse for all $\mat x \in \F_2^{t\log(m/t)}$.
\end{lemma}

\begin{proof}
    By union bound, it suffices to show that for any fixed $\mat x \in \F_2^{t\log(m/t)}$, letting $\widetilde{\mat x} = \spfy_{m,t}(\mat x)$, we have:
    \[\Pr\left[ \wt(\mat E \cdot \widetilde{\mat x}) > \gamma\ell \;\middle|\; \mat E \gets \Ber(\F_2,\epsilon)^{\ell \times m}\right] \le 2^{-t\log(m/t)} \cdot \negl(\lambda).\]
    Let $\mat e = \mat E \cdot \widetilde{\mat x}$. We can see that for each $j \in [\ell]$, the entry $(\mat E \cdot \widetilde{\mat x})_j = \langle \mat E_j,\widetilde{\mat x} \rangle$, where $\mat E_j$ is the $j^{th}$ row of $\mat E$, is drawn from the Bernoulli distribution with noise 
    \begin{align*}
        \epsilon' = \frac{1 - (1-2\epsilon)^t}{2} < \epsilon t = \frac{\gamma}{\alpha+1},\qquad\text{by \Cref{lem:bias-bernoulli} and the choice of $\epsilon$}.
    \end{align*}
    We now apply Chernoff bound (\Cref{lem:chernoff-bound}) to get
    \begin{align*}
        \Pr\left[\wt(\mat E \cdot \widetilde{\mat x}) > \gamma \ell \;\middle|\; \mat E \gets \Ber(\F_2,\epsilon)^{\ell \times m}\right] 
        &\le e^{-2 \alpha^2\frac{\gamma}{\alpha+1} \ell} = e^{-2 \frac{\alpha^2}{\alpha+1}\frac{\gamma}{\rho_{\calC}} t\log(m/t)} < 2^{-2t\log(m/t)}.
    \end{align*}
    In the above, the last inequality follows from our choice of $\alpha$.
    Therefore, $\mat E \cdot \widetilde{\mat x}$ has Hamming weight at most $\gamma \ell$ for all $\mat x \in \F_2^{t\log(m/t)}$ with probability at least \[1 - 2^{t\log(m/t)} \cdot 2^{-2t\log(m/t)} = 1 - 2^{-t\log(m/t)} = 1 - \negl(\lambda).\]
\end{proof}

\begin{theorem}\label{thm:ltdf}
    Assuming the $(n,m,k,\calM_{\sparse},\epsilon)$-Dense-Sparse LPN assumption, where $n,m,k,\calM_{\sparse},\epsilon$ are chosen as in \Cref{fig:LTDF}, the construction in \Cref{fig:LTDF} is an $\ABOLTDF$ with input length $t\log(m/t)$, branches indexed by $\F_2^{t\log(m/t)}$, and lossiness factor $\Gamma > 1$.
\end{theorem}

\begin{proof}
    We argue each property separately as follows.
    
    \vspace{.5em}

    \emph{Mode indistinguishability.} This is clear from the Dense-Sparse LPN assumption, since in both cases, $\fk=(\mat A,\mat B)$ is indistinguishable from $(\mat A,\mat U)$ for a random $\mat U \gets \F_2^{\ell \times m}$.

    \vspace{.5em}
    
    \emph{Trapdoor inversion.} If $\Gen(1^\lambda,\inj) \to (\fk=(\mat A,\mat B),\td=(\mat S, \tau^*))$ and $F(\fk,\tau,\mat x) = (\mat y_1, \mat y_2)$ for some $\mat x \in \F_2^{t\log(m/t)}$, letting $\widetilde{\mat x} = \spfy_{m,t}(\mat x)$, we have that 
    \begin{align*}
        \mat y_1 = \mat A \cdot \widetilde{\mat x}, \quad
        \mat y_2 &= (\mat B - \mat C^\top \cdot \mat H_{\tau}^\top \cdot \mat G_{m,t}^\top) \cdot \widetilde{\mat x} \\
        &= \mat S \cdot \mat y_1 + \mat E \cdot \widetilde{\mat x} + \mat C^\top \cdot \mat H_{\tau^* - \tau} \cdot \mat x.
    \end{align*}
    The last equality follows from \Cref{eqn:gadget-inversion}, namely that $\mat G_{m,t}^\top \cdot \spfy_{m,t}(\mat x) = \mat x$. Hence for all $\tau \ne \tau^*$, with all but $\negl(\lambda)$ probability, we have:
    \begin{align*} 
        F^{-1}(\td,(\mat y_1,\mat y_2)) &= \left(\mat H_{\tau^* - \tau}\right)^{-1} \cdot\Decode\left(\mat y_2 - \mat S \cdot \mat y_1\right) \\
        &= \left(\mat H_{\tau^* - \tau}\right)^{-1} \cdot \Decode\left(\mat C^\top \cdot \mat H_{\tau^* - \tau} \cdot \mat x + \mat E \cdot \spfy_{m,t}(\mat x)\right)\\
        &= \left(\mat H_{\tau^* - \tau}\right)^{-1} \cdot \left(\mat H_{\tau^* - \tau} \cdot \mat x\right) \qquad\text{(holds with $1-\negl(\lambda)$ probability)}\\
        &= \mat x.
    \end{align*}
    Indeed, the third equality holds with probability $1-\negl(\lambda)$ because of the following. First, the fact that $\Decode$ can efficiently decode from any $\gamma \ell$ errors, i.e. $\Decode(\mat C^\top \cdot \mat x + \mat e) = \mat x$ for any $\mat e \in \calB_{\le}(\ell,\gamma \ell)$. Second, from \Cref{lem:ltdf-noise-bound} and the choice of $\gamma$, we know that $\wt(\mat E \cdot \spfy_{m,t}(\mat x)) \le \gamma \ell \le \delta_{\calC}\ell$ for any $\mat x \in \F_2^{t\log(m/t)}$, which holds with all but $\negl(\lambda)$ probability.

    \vspace{.5em}

    \emph{Lossyness.} For the lossy branch $\tau=\tau^*$, we will count the number of attainable values $\mat y=(\mat y_1,\mat y_2) \in \F_2^{n/2+\ell}$ in the range of $F$. Denote $\calX = \F_2^{t \log(m/t)}$ to be the domain, and $\calY = \{(\mat y_1,\mat y_2)\} \subset \F_2^{n/2+\ell}$ to be the range.
    
    The first part $\mat y_1 := \mat A \cdot \widetilde{\mat x} = \mat T \cdot (\mat M \cdot \widetilde{\mat x})$, is determined by the value $\mat M \cdot \widetilde{\mat x} \in \calB(n,\le kt)$, hence has cardinality at most $\abs{\Ball_{\le}(n,kt)}$. By our choice of parameters for $n,m,t$ so that \Cref{lem:compress} is satisfied with compression factor $D$, it follows that the set $\calY_1$ of possible values of $\mat y_1$ satisfies $\abs{\calY_1} < \abs{\calX}^{1/D} = 2^{t\log(m/t)/D}$.

    Next, for a fixed $\mat y_1$, using \Cref{lem:ltdf-noise-bound}, we have that 
    \[\mat y_2 = \mat B_{\tau} \cdot \widetilde{\mat x} = \mat S \cdot \mat y_1 + \mat E \cdot \widetilde{\mat x}\] is in a Hamming ball of radius $\gamma \ell$ around $\mat y_1$ except with $\negl(\lambda)$ probability. Thus we have that
    \begin{align*}
        \abs{\calY} &\le \abs{\calY_1} \cdot \abs{\calB_{\le}(\ell,\gamma \ell)}\\ &\le 2^{t\log(m/t) / D} \cdot 2^{H(\gamma) \ell} \\
        &\le \left(2^{t\log(m/t)}\right)^{1/D+1/D'}\\
        &= \abs{\calX}^{1/\Gamma},
    \end{align*}
    and thus the lossiness parameter is at least $\Gamma$. Here the second inequality is due to upper bounds on $\abs{\calY_1}$ and $\abs{\calB_{\le}(\ell,\gamma \ell)}$, the third inequality is due to our choice of $\gamma \le H^{-1}(\rho/D')$, which implies that $H(\gamma) \ell = \frac{H(\gamma)}{\rho} \cdot t\log(m/t) \le \frac{1}{D'} \cdot t\log(m/t)$, and the final equality is because $\frac{1}{D}+\frac{1}{D'}=\frac{1}{\Gamma}$.
\end{proof}
\section{Cryptanalysis on Dense-Sparse LPN}\label{sec:cryptanalysis}

Recent works, in the context of generating correlated pseudorandomness for MPC applications \cite{CCS:BCGI18,C:BCGIKS19,FOCS:BCGIKS20}, have proposed various novel variants of LPN with different matrix distributions~\cite{FOCS:BCGIKS20,C:CouRinRag21,C:BCGIKRS22,PKC:CouDuc23,C:RagRinTan23,C:BCCD23}. These recent progress also came with a \emph{systematization} of known attacks on LPN-style assumption. Namely, prior works observed that most attacks (Information-Set Decoding~\cite{prange-isd}, BKW~\cite{STOC:BluKalWas00}, Gaussian Elimination~\cite{C:EssKubMay17}, Statistical Decoding~\cite{jabri-stat-decoding}, etc.) on LPN-style assumptions can be captured by a unified framework for the security of LPN variants, namely the \emph{linear test framework}, which we will also use to analyze the security of our assumption. The linear test framework has also been used to extensively cryptanalyze Goldreich's PRGs~\cite{Gol00,CM01,FOCS:MosShpTre03,TCC:CEMT09,APPROX:BogQia09,STOC:AppBarWig10,TCC:AppBogRos12,BogdanovQ12,STOC:Applebaum12,SIAM:App13,CCC:OdoWit14,STOC:AppLov16,STOC:KMOW17,AC:CDMRR18,FOCS:AppKac19}, which is a related assumption to Sparse LPN.

In a linear test, an adversary takes as input a matrix $\mat{A}\in \Z^{n\times m}_{2}$ and outputs a vector $\mat{v}\in \Z^{m\times 1}_{2}$ such that $(\mat{s} \mat{A}+\mat{e}) \,\mat{v}$ is a biased random variable with an inverse polynomial bias towards $0$. Assuming $\mat{e}$ is chosen from Bernoulli distribution with probability $\epsilon$, a successful attacker must output $\mat{v}$ with hamming weight $O(\frac{\log n}{\epsilon})$, or else the bias is negligible.

\begin{definition}[Security against Linear Test, adapted from \cite{C:CouRinRag21}]\label{def:linear-test-framework}
    Let $n \in \N$ be the dimension, $m=m(n)$ be the number of samples, and $q=q(n)$ be a prime power. Given an efficiently sampleable distribution $\calM = \calM(n,m,\F_q)$ over matrices in $\F_q^{n \times m}$ and noise probability $\epsilon=\epsilon(n) \in (0,1)$, we say that the $(\calM,\epsilon)\text{-}\LPN$ assumption is \emph{$(T(n),\beta(n),\delta(n))$-computationally secure against linear tests} if for any adversary $\advA$ running in time at most $T(n)$, it holds that
    \[\Pr\left[\bias_{\mat v}(\calD_{\mat A}) \ge \beta \quad \middle| \quad \begin{aligned}
        &\mat A \gets \calM \\
        &\mat v \gets \advA(\mat A)
    \end{aligned}\right] \le \delta,\]
    where $\calD_{\mat A}=\left\{\mat s \mat A + \mat e \mid \mat s \sample \F_q^{1 \times n}, \mat e \gets \Ber(\F_q,\epsilon)^{1 \times m}\right\}$. (Recall that bias is defined in \Cref{def:bias}.)
\end{definition}







In other words, the Linear Test Hypothesis for the $(\calM,\epsilon)$-LPN assumption states that $(\calM,\epsilon)$-LPN is secure if it is computationally difficult to find $O(\frac{\log n}{\epsilon})$-sparse vectors in the kernel of $\mat A \gets \calM$. All known counter-examples to this hypothesis are for distributions $\calM$ of algebraically structured matrices (see~\cite{C:BCGIKS20} for a detailed discussion). In contrast, our Dense-Sparse matrix distribution $\mat A = \mat T \mat M$ has no apparent algebraic structure, and thus it is reasonable to conjecture that Dense-Sparse LPN is secure assuming it is secure against linear tests. 
Furthermore, if finding $O(\frac{\log n}{\epsilon})$-sparse vectors in the kernel is difficult for subexponential time adversaries, then the corresponding LPN assumption is also subexponentially secure.

Assuming the Linear Test Hypothesis, we now examine the hardness of finding sparse vectors in the kernel of $\mat A = \mat T \mat M$. Such vectors of sparsity $n^{\delta}$ for any $\delta \in (0,1)$ actually come from the kernel of $\mat M$, with all but negligible probability. This is because $\mat T$ is a random binary matrix of dimension $\frac{n}{2} \times n$, and so typically only has vectors in the kernel that are at least $\Omega(n/\log n)$-sparse (see~\cite{C:CouRinRag21} for a detailed calculation). Therefore, we will focus our attention on the size of short vectors in the kernel of the $k$-sparse matrix $\mat M$, and also on the overall security of Sparse LPN.

\subsection{Security of Sparse LPN}

Sparse LPN is a relatively well-understood assumption in the domain of refutations for random constraint satisfaction. When one samples a random matrix for the sparse LPN assumption $\mat{M}$ with sparsity parameter $k$, dimension $n$ and sample complexity $n^{1+(\frac{k}{2}-1)(1-\delta)}$, it can be proven that the dual distance of $\mat{M}$, chosen randomly from $\SpMat(n,m,k,\F)$ (the set of $k$-sparse matrices in $\F^{n \times m}$), is $t=\Theta(n^{\delta})$ with all but inverse polynomial probability.

\begin{lemma}[Folklore, see \cite{C:DIJL23}]\label{lem:large-dual-distance}
    For any finite field $\F$, given any $k=k(n) \ge 3$, any $0 < \delta < 1$, and $m=O\left(n^{1+(\frac{k}{2}-1)(1-\delta)}\right)$, there exists a constant $c > 0$ such that the following holds for large enough $n$:
    \begin{align*}
        \Pr\left[\dd(\mat M) \le c \cdot n^\delta \;\middle|\; \mat M \gets \mathsf{SpMat}(n,m,k,\F)\right] = \Theta\left( \left(\frac{k}{n^{\delta}}\right)^{k-2} \right).
    \end{align*}
\end{lemma}
Note that the dual distance of $t=O(n^{\delta})$ is also known in the stronger case of \emph{worst-case} $k$-sparse $\mat M$, whenever $m = \tilde{O}\left(n^{1+(\frac{k}{2}-1)(1-\delta)}\right)$~\cite{guruswami2022algorithms,hsieh2023simple}.

One could then consider Sparse LPN in two regimes. In the first regime, the error probability is  $\epsilon=\omega(\frac{\log n}{t})$. In this regime, conditioned on the event that there simply \emph{does not exist} vectors of sparsity $O(\frac{\log n}{\epsilon})$ in the kernel of $\mat M$, security against linear tests holds. Sparse LPN with this regime of error (even constant error probability) has found prior applications~\cite{STOC:IKOS08,C:ADINZ17,EC:AppKon23,EC:BCGIKRS23,C:DIJL23}.

Nevertheless, there are also applications of Sparse LPN where sparse vectors (in the kernel of $\mat M$) do exist, but seem computationally hard to find. One such example is the public-key encryption scheme of Applebaum, Barak and Wigderson~\cite{STOC:AppBarWig10} uses an inverse polynomial error probability $\epsilon = O(\frac{\log n}{t})$. In this case, certainly there exist lots of $t$-sparse vectors in the kernel of $\mat{M}$. To this date, we do not know any procedure that can efficiently find (even in subexponential time $2^{t^\rho}$ for any $0 < \rho < 1$) such $t$-sparse vectors in the kernel of $\mat{M}$. Therefore, the linear test framework correctly predicts the (subexponential) security of this assumption.

Our applications could have been based on Sparse LPN in its compression regime. This holds at a significantly lower error probability $\epsilon_{\mathsf{cps}}=o(\frac{1}{t_{\mathsf{cps}}})$ where $t_{\mathsf{cps}} \sim \left(\frac{n^{k}}{m}\right)^{\frac{1}{k-1}}$, which is polynomially larger than the dual distance $t \sim \left(\frac{n^{k/2}}{m}\right)^{\frac{1}{k/2-1}}$. We in fact show that it is \emph{easy} to find $t_{\mathsf{cps}}$-sparse solutions to $\mat{M}\,\mat{x}$ for randomly chosen $\mat{M}$.

\paragraph{Sparse LPN is Broken in its Compression Regime.}  Following our exposition in \Cref{sec:tech-overview}, we give a simple attack on Sparse LPN in the parameter regime needed for achieving compression (and hence lossy trapdoor functions).

\begin{theorem}
\label{thm:sparse-lpn-attack}
    Consider Sparse LPN with sparsity $k$, dimension $n$ and sample complexity $m$ in the regime in Lemma~\ref{lem:compress}, with error rate $\epsilon = O(1/t)$. Assuming the matrix is chosen by sampling distinct $k$-sparse random columns, this variant of Sparse LPN can be broken in polynomial time.
\end{theorem}

The attack, given samples $(\mat M, \mat b)$ with $\mat b$ either from the Sparse LPN distribution, or the random distribution:
\begin{enumerate}
    \item Pick a random subset $S \subset [n]$ of size $t$. Initialize $T=\emptyset$.
    \item For each column $j \in [m]$ of $\mat M$, if all $k$ of the non-zero entries lie inside $S$, add $j$ to $T$.
    \item If $\abs{T} > t$, find a linear dependency (expressed as a vector $\mat x$) between the columns in $T$.
    \item Compute $d = \langle \mat b, \mat x\rangle$. If $d=0$, return ``Sparse LPN'', else return ``random''.
\end{enumerate}

\begin{lemma}\label{lem:sparse-lpn-attack-analysis}
    This attack succeeds with high probability $1-o(1)$ whenever $m > c \cdot m_{(\ref{lem:compress})}(k,D)$ for a sufficiently large constant $c$ and any compression factor $D>1$.
\end{lemma}

\begin{proof}
    For any fixed choice of $S \subset [n]$, we can compute the expectation of the number of columns found in step 2 as follows.
    \begin{align*}
        \E[\abs{T}] = m \cdot \Pr[\Supp(\mat a) \subset S \mid \mat a \gets \SpMat(n,1,k)] = m \cdot \frac{{t \choose k}}{{n \choose k}} \approx m \cdot \left(\frac{t}{n}\right)^k > t.
    \end{align*}
    The last inequality is due to the choice $m > m_{(\ref{lem:compress})}(k,1)$.
    
\end{proof}

\subsection{Dual-Distance of Dense-Sparse LPN}

We now examine the dual-distance of the matrices arising in our Dense-Sparse LPN assumption, and the computational hardness of finding sparse vectors $\vec x$ that are in the kernel of those matrices. Since our matrices are of the form $\mat{T}\mat{M}$ where $\mat{T} \in \Z^{n\times m}_{2}$ is randomly chosen matrix with high probability the dual distance is around $t=O(n^{\delta})$ where $m=n^{1+(\frac{k}{2}-1)(1-\delta)}$ where as the error probability $\epsilon <\frac{1}{t_{\mathsf{cps}}}$, where $t_{\mathsf{cps}} \sim \left(\frac{n^{k}}{m}\right)^{\frac{1}{k-1}}$ is the threshold for compression.

\paragraph{Sparse Vector in the Kernel of $\mat{T}\mat{M}$.} 
As described above, if we can efficiently find $t=O(n^{\delta})$ sparse vectors $\vec x\in \Z^{m\times 1}_{2}$ in the kernel of $\mat{T}\mat{M}$, that is enough to break our assumption with any error probability $O(\frac{\log n}{t})$. In our specific setting, we work with an error probability of about $\frac{1}{t_{\mathsf{cps}}}$ where $t_{\mathsf{cps}} \sim \left(\frac{n^{k}}{m}\right)^{\frac{1}{k-1}} \approx \frac{n}{m^{1/k}}$ (for large enough constant $k$). Setting $m=n^{1+(k/2-1)(1-\delta)}$, this turns out to be roughly $t_{\mathsf{cps}}\approx n^{\frac{1}{2}+\delta-\frac{1}{k}}$. On the other hand, much sparser vectors ($t=O(n^{\delta})$ sparse) in the kernel exist. Thus, if one can find kernel vectors of this sparsity, it could break our assumption. However, for arbitrary dense matrices we don't know a better way to find such $n^{\delta}$-sparse vectors than naive search, which is known to be subexponentially hard even given $\mat{M}$ in the clear (this very assumption was made by \cite{STOC:AppBarWig10}).

It is also true that one could also break our assumption by finding vectors in the kernel with much higher sparsity $t_{\mathsf{cps}}\approx n^{\frac{1}{2}+\delta-\frac{1}{k}}$. In fact, we showed how to do find such vectors in~\Cref{thm:sparse-lpn-attack}, but our attack crucially relies on seeing the \emph{sparsity pattern} of the matrix $\mat{M}$. This pattern is no longer apparent when we only give out $\mat A = \mat T \mat M$ with $\mat T$ being a random (dense) matrix, so the attack in~\Cref{thm:sparse-lpn-attack} no longer applies.

\paragraph{Pseudorandomness of $\mat{T}\mat{M}$.} We in fact postulate that the matrices $\mat{T}\mat{M}$ are computationally hard to distinguish from random matrices (which have dual distance of $O(\frac{n}{\log n})$ with overwhelming probability). A first observation is that since any $O(t)$ columns of $\mat{M}$ are linearly independent (as its dual distance is  $\Omega(t)$) with overwhelming probability (assuming $\mat{M}$ is chosen from a \hyperref[def:good-dist]{good} distribution), any $O(t)$ columns of $\mat{T}\mat{M}$ are distributed according to a random distribution over $\Z^{\frac{n}{2}\times 1}_{2}$. This ensures that matrix is from a $O(t)$-wise independent distribution.

\subsubsection{Searching for \texorpdfstring{$\mat{T}$}{T}}
Next, we examine the possibility of learning $\mat T$ from the matrix $\mat T \mat{M}$ and then finding $t_{\mathsf{cps}}$ sparse vectors in the kernel of $\mat{M}$. We provide an attack that works with $\mat{T}$ is an invertible matrix in $\Z^{n\times n}_{2}$. This attack does not apply when the matrix lives inside $\Z^{\beta n\times n}_{2}$ for any constant $\beta \in (0,1)$, which justifies our choice of parameters for Dense-Sparse LPN.

\paragraph{When $\mat T \gets \F_2^{n \times n}$ is a square matrix.}
We can generalize our previous attack to rule out the case where $\mat T$ is a random square matrix (and not just the identity). The idea is that this $\mat T$ may be ``unmasked'' by searching for its inverse $\mat Z$. In other words, given the Dense-Sparse matrix $\mat A = \mat T \cdot \mat M$, we want to find $\mat Z \in \F_2^{n \times n}$ such that $\mat Z \cdot \mat A = \mat M$ is a sparse matrix. Let a candidate row for $\mat{Z}$ be $\mat{z}$. Observe that $\mat{z}\mat{A}$ must be a sparse vector that corresponds to the row of $\mat{M}$. Since each column of $\mat{M}$ has some constant number $k$ of elements, with high probability (roughly $1-O(\frac{k}{n})$) each coordinate of $\mat{z}\mat{A}$ must be $0$. Then, we can essentially apply the same attack that breaks LPN assumption with error probability $O(\frac{1}{n})$ to learn $\mat{z}$.

The attack is as follows. Let the columns of $\mat{A}$ be $\mat{a}_1,\ldots,\mat{a}_m$. We randomly sample $n$ equations at random, say denoted by $I$,$\{\mat{a}_i\}_{i\in I}$. There is probability $(1-\frac{k}{n})^{n}=\Omega(1)$ chance that there exist $\vec z$ that satisfy $\langle \vec z, \vec a_i\rangle=0$. We can then find all such vectors $\vec z$ that will form the rows of $\vec Z$. This can be used derive $\mat{M}$ up to a permutation of rows, which is sufficient for carrying out the attack on Sparse LPN. 


\paragraph{When $\mat T \gets \F_2^{\beta n \times n}$ is a compressing matrix for $\beta<1$.} A natural way to extend the above attack would be to find an inverse $\mat{Z}$ for a square sub-matrix of $\mat T$. For example, we can hope to find $\mat{Z}\in \Z^{\beta n\times \beta n}_{2}$ that is the inverse of the submatrix formed by the first $\beta\cdot n$ sub-columns of $\mat{T}$, so that $\mat{Z}\mat{T}\mat{M}=[\mat{I}_{\beta n} \Vert \mat{T}']\cdot \mat{M}$. 
Note that $\mat{T}' \in \F_2^{\beta n \times (1-\beta) n}$ is randomly distributed.

Now, if one examines $[\mat{I}_{\beta n} \Vert \mat{T}']\cdot \mat{M}$ the following holds. For any $i$-th column $\mat{m}_i$ of $\mat{M}$, if $\mat{m}_i$ is supported over the first $\beta \cdot n$ variables, then the same $i$-th column in $[\mat{I}_{\beta n} \Vert \mat{T}']\cdot \mat{M}$ is precisely $\mat m_i$, and hence remains $k$-sparse. Otherwise, if $\mat m_i$ has a non-zero entry in the rest of the positions, then the randomness of $\mat{T}'$ would make the resulting column in $\mat Z \mat T \mat M$ random as well.

Let's calculate when the first event happens, so that we have hope of unmasking a column of $\mat{M}$. Since there is a constant chance that when choosing a column of $\mat{M}$, that column is supported inside the first $\beta\cdot n$ set of variables, the product $[\mat{I}_{\beta n} \Vert \mat{T}']\cdot \mat{M}$ will be so that each row has a constant fraction of non-zero element. As a consequence, to solve for $\mat{Z}$ we might have to rely on an LPN solver that works with constant probability noise which might be hard to do. In general, this attack should take (near-)exponential time as long $\mat{T}$ is randomly chosen from $\Z^{\beta n\times n}_{2}$ for any $\beta >0$.

\subsubsection{On Subexponential Time Attacks}

We have discussed two main attacks for finding sparse vectors in the kernel of the Dense-Sparse matrix $\mat T \mat M$, which by the linear test framework (\Cref{def:linear-test-framework}) is the best attack against Dense-Sparse LPN. In the fist attack, we attempt to directly find $O(n^{\delta})$ sparse vectors in the kernel of $\mat T \mat M$, treating it as an arbitrary dense matrix. In the second attack, we attempt to peel off $\mat{T}$, and then use our attack in~\Cref{thm:sparse-lpn-attack} to find sparse vectors in the kernel of $\mat M$. As argued above, the best known algorithms for the first takes $2^{\tilde{O}(n^{\delta})}$ time and the second to the best of our current knowledge should take almost exponential time assuming $\beta$ is some constant. Therefore, we may justify our conjecture that Dense-Sparse LPN is hard against subexponential time adversaries.
\pagebreak

{\small
\bibliographystyle{alpha}
\bibliography{Bibliography/abbrev3,Bibliography/crypto,Bibliography/custom}}




\end{document}